\documentclass[a4paper,twocolumn,11pt,accepted=2023-12-26]{quantumarticle}

\pdfoutput=1
\usepackage[utf8]{inputenc}
\usepackage[english]{babel}
\usepackage{hyperref}

\usepackage{graphicx}
\usepackage{epstopdf}
\usepackage{amsmath}
\usepackage{amssymb}
\usepackage{mathtools}
\usepackage{mathrsfs}
\usepackage{amsthm}
\usepackage{bm}
\usepackage{bbm}
\usepackage{url}
\usepackage[T1]{fontenc}
\usepackage{csquotes}
\MakeOuterQuote{"}
\usepackage{thmtools, thm-restate}

\usepackage[square,numbers,sort&compress]{natbib}

\newtheorem{thm}{Theorem}
\newtheorem{lem}{Lemma}

\newtheorem{proposition}{Proposition}

\theoremstyle{definition}

\theoremstyle{remark}

\def\vec#1{{\bm{#1}}} 

\newcommand{\tr}{\operatorname{tr}}

\newcommand{\imply}{\mathrel{\Rightarrow}}

\newcommand{\poly}{\mathrm{poly}}

\newcommand{\rmd}{\mathrm{d}}

\newcommand{\caH}{\mathcal{H}}

\newcommand{\caV}{\mathcal{V}}

\newcommand{\sym}{\mathrm{sym}}

\newcommand{\scrA}{\mathscr{A}}

\newcommand{\scrM}{\mathscr{M}}

\newcommand{\lsp}{\hspace{0.1em}}

\newcommand{\be}{\begin{equation}}
\newcommand{\ee}{\end{equation}}
\newcommand{\ba}{\begin{align}}
\newcommand{\ea}{\end{align}}

\def\<{\langle}  
\def\>{\rangle}  
\newcommand{\ket}[1]{| #1\>}
\newcommand{\bra}[1]{\< #1|}

\def\eqref#1{\textup{(\ref{#1})}}  
\newcommand{\eref}[1]{Eq.~\textup{(\ref{#1})}}
\newcommand{\Eref}[1]{Equation~\textup{(\ref{#1})}}

\newcommand{\esref}[2]{Eqs.~\textup{(\ref{#1})} and \textup{(\ref{#2})}}

\newcommand{\Esref}[2]{Equations~\textup{(\ref{#1})} and \textup{(\ref{#2})}}

\newcommand{\fref}[1]{Fig.~\ref{#1}}

\newcommand{\sref}[1]{Sec.~\ref{#1}}
\newcommand{\Sref}[1]{Section~\ref{#1}}

\newcommand{\thref}[1]{Theorem~\ref{#1}}

\newcommand{\thsref}[1]{Theorems~\ref{#1}}

\newcommand{\lref}[1]{Lemma~\ref{#1}}

\newcommand{\lsref}[1]{Lemmas~\ref{#1}}

\newcommand{\pref}[1]{Proposition~\ref{#1}}

\newcommand{\cref}[1]{Conjecture~\ref{#1}}
\newcommand{\Cref}[1]{Conjecture~\ref{#1}}

\newcommand{\aref}[1]{Appendix~\ref{#1}}

\newcommand{\rcite}[1]{Ref.~\cite{#1}}
\newcommand{\rscite}[1]{Refs.~\cite{#1}}

\begin{document}
\title{Efficient Verification of Ground States of Frustration-Free Hamiltonians}

\author{Huangjun Zhu}
\email{zhuhuangjun@fudan.edu.cn}
\orcid{0000-0001-7257-0764}
\affiliation{State Key Laboratory of Surface Physics and Department of Physics, Fudan University, Shanghai 200433, China}
\affiliation{Institute for Nanoelectronic Devices and Quantum Computing, Fudan University, Shanghai 200433, China}
\affiliation{Center for Field Theory and Particle Physics, Fudan University, Shanghai 200433, China}

\author{Yunting Li}
\affiliation{State Key Laboratory of Surface Physics and Department of Physics, Fudan University, Shanghai 200433, China}
\affiliation{Institute for Nanoelectronic Devices and Quantum Computing, Fudan University, Shanghai 200433, China}
\affiliation{Center for Field Theory and Particle Physics, Fudan University, Shanghai 200433, China}

\author{Tianyi Chen}
\affiliation{State Key Laboratory of Surface Physics and Department of Physics, Fudan University, Shanghai 200433, China}
\affiliation{Institute for Nanoelectronic Devices and Quantum Computing, Fudan University, Shanghai 200433, China}
\affiliation{Center for Field Theory and Particle Physics, Fudan University, Shanghai 200433, China}


\begin{abstract}
Ground states of local Hamiltonians are of key interest in many-body physics and also in quantum information processing. Efficient verification of these states are crucial to many applications, but very challenging. Here we propose a simple, but powerful  recipe for verifying the ground states of general frustration-free Hamiltonians based on local measurements. Moreover, we  derive rigorous bounds on the sample complexity by virtue of the quantum detectability lemma (with improvement) and quantum union bound. Notably, the number of samples required does not increase with the system size when the underlying Hamiltonian is local and gapped, which is the case of most interest. As an application, we propose a general approach  for verifying Affleck-Kennedy-Lieb-Tasaki (AKLT) states on arbitrary graphs based on local spin measurements, which requires only a constant number of samples for AKLT states defined on various lattices. 
Our work is of interest not only to many tasks in quantum information processing,  but also to the study of many-body physics. 
\end{abstract}

\maketitle


\section{Introduction}
Multipartite entangled states play key roles in various quantum information processing tasks, including quantum computation, quantum simulation, quantum metrology, and quantum networking. An important class of multipartite states are the ground states of local Hamiltonians,
such as  Affleck-Kennedy-Lieb-Tasaki (AKLT) states \cite{AfflKLT87,AfflKLT88} and many tensor-network states
\cite{PereVWC08,CiraPSV21}.  They are of central interest in traditional condensed matter physics and also in the recent study of symmetry-protected   topological orders \cite{ChenGLW12, Sent15,ChiuTSA16,WeiRA22}. In addition, such states are particularly appealing to quantum information processing because they
can be prepared by cooling down  \cite{VersWC09} and adiabatic evolution \cite{FarhGGS00,FarhGGL01,AlbaL18, GeMC16,CruzBTS22} in addition to quantum circuits. Recently, they have found increasing applications in quantum computation and simulation \cite{StepWPW17, RausOWS19, StepNBE19,DaniAM20,GoihWET20,HangE23,BermHSR18}. Notably, AKLT states on various 2D lattices, including the honeycomb lattice, can realize universal  measurement-based quantum computation \cite{KaltLZB10,WeiAR11,Miya11,WeiAR12,Wei18,WeiRA22}.

To achieve success in quantum information processing, it is crucial to guarantee that the underlying multipartite quantum states satisfy desired requirements,  irrespective of whether they are prepared by quantum circuits or as ground states of local Hamiltonians. Unfortunately, traditional tomographic methods are notoriously resource consuming. To resolve this problem, many researchers have tried to find more efficient 
alternatives \cite{EiseHWR20,CarrEKK21,KlieR21,YuSG22,MorrSGD22,HangE23}. Recently, a powerful approach known as \emph{quantum state verification} (QSV) has attracted increasing attention \cite{HayaMT06,CramPFS10,AoliGKE15,LanyMHB17,HangKSE17, PallLM18,TakeM18,ZhuH19AdS,ZhuH19AdL,WuBGL21,LiuSHZ21,GocaSD22}. Efficient verification protocols based on local measurements have been constructed for many states of practical interest, including bipartite pure states \cite{HayaMT06,Haya09G,PallLM18,ZhuH19O,LiHZ19,WangH19,YuSG19}, stabilizer states  \cite{HayaM15,FujiH17,HayaH18, PallLM18,ZhuH19E,ZhuH19AdL,LiHZ20,MarkK20,LiZH23}, hypergraph states \cite{ZhuH19E}, weighted graph states \cite{HayaT19}, and Dicke states \cite{LiuYSZ19,LiHSS21}. The efficiency of this approach has been demonstrated in quite a few experiments \cite{ZhanZCP20,ZhanLYP20,LuXCC20,JianWQC20}.
Although several works have considered the verification of ground states of local Hamiltonians \cite{CramPFS10,LanyMHB17,HangKSE17,BermHSR18,GluzKEA18,TakeM18,CruzBTS22}, the sample costs of known protocols are still too prohibitive for large and intermediate quantum systems of practical interest, which consist of more than 100 qubits or qudits, even if the   Hamiltonians are frustration free.

In this work, we propose a general  recipe for verifying the ground states of  frustration-free Hamiltonians based on local measurements, which does not require explicit expressions for the ground states. Each  protocol is constructed from a matching cover or edge coloring of a  hypergraph encoding the action of the Hamiltonian and is thus very intuitive. Moreover, we derive a rigorous performance guarantee by virtue of the spectral gap of the underlying Hamiltonian and simple graph theoretic quantities, such as the \emph{degree} and \emph{chromatic index} (also known as edge chromatic number). For a local  Hamiltonian defined on a lattice, to verify the ground state within  infidelity $\epsilon$ and significance level $\delta$, the sample cost is only $O((\ln \delta^{-1})/(\gamma\epsilon))$, where $\gamma$ is the spectral gap  of the underlying Hamiltonian. Compared with previous protocols, the scaling behaviors are much better   with respect to the system size, spectral gap $\gamma$, and the precision as quantified by the infidelity $\epsilon$.  Notably, we can verify the ground state  with a constant sample cost that is independent of the system size when the spectral gap $\gamma$ is bounded from below by a positive constant.

For example, we can  verify AKLT states defined on arbitrary graphs, and the resource overhead is upper bounded by a constant that is independent of the system size for most AKLT states of practical interest, including those defined on various 1D and 2D lattices. When  $\epsilon=\delta=0.01$ and the number of nodes is around 100, our protocols are  tens of thousands of times more efficient than previous protocols, and the efficiency advantage becomes more and more significant as the target precision and system size increase. Additional details can be found in the companion paper~\cite{ChenLZ23} (see \aref{app:Companion}). 
 Our recipe is expected to find diverse applications in quantum information processing and many-body physics.

Technically, a frustration-free Hamiltonian is relatively easy to deal with because such a  Hamiltonian offers the  possibility of parallel verification of multiple local terms simultaneously.
High degree of parallel processing is crucial to achieving a high efficiency. 
Nevertheless, it is still  nontrivial to combine the local information as efficiently as possible
so as to draw an accurate conclusion about the ground state of the whole Hamiltonian, although the basic idea is quite intuitive. Our main contribution in this work is to
construct nearly optimal verification protocols and to derive nearly tight bounds on the sample complexity by fully exploiting the frustration-free property.

To establish our main results on the sample complexity, we first derive a stronger quantum detectability lemma, which improves over previous results \cite{AharALV09,AnshAV16}. This stronger detectability lemma can be used to derive tighter upper bounds on the operator norm of certain product of  projectors  tied to the projectors that compose the Hamiltonian; it is also of general interest that is beyond the focus of this paper. In addition, by virtue of a generalization of the quantum union bound \cite{Gao15, DonnV22}, we derive a simple, but nearly tight lower bound on the spectral gap  of the average of certain test projectors based on the operator norm of a product. Combining these technical tools we derive our main result \thref{thm:SpectralGap}, which clarifies the spectral gap and sample complexity of our verification protocol. Other main results (including \thsref{thm:SpectralGap2} and \ref{thm:homoDesignGen}), once formulated, can be proved by virtue of standard arguments widely used in quantum information theory. Familiarity with the concepts of spherical $t$-designs \cite{DelsGS77,Seid01,BannB09} and spin coherent states \cite{ZhanFG90} would be helpful to understanding the proof of \thref{thm:homoDesignGen}.

The rest of this paper is organized as follows. In \sref{sec:Pre} we first review the basic framework of quantum state verification and discuss its generalization to subspace verification. 
Then we introduce basic concepts on hypergraphs and frustration-free Hamiltonians that are relevant to the current study. In \sref{sec:DL} we prove a stronger detectability lemma and discuss its implications. In \sref{sec:EVGS} we propose a general approach for verifying the ground states of frustration-free Hamiltonians and determine the sample complexity. In \sref{sec:AKLT} we illustrate the power of this simple idea by constructing efficient protocols for verifying general AKLT states. \Sref{sec:sum} summarizes this paper.

\section{\label{sec:Pre}Preliminaries}
\subsection{Quantum state verification}
Consider a device that is supposed to produce the target state $|\Psi\rangle$ within the Hilbert space $\caH$, but actually produces the states $\sigma_1, \sigma_2, \ldots, \sigma_N$ in $N$ runs. Our task is to verify whether these states are sufficiently close to the target state on average, where the closeness is usually quantified by the fidelity. To this end, in each run we can perform a random test from a set of accessible tests. Each test is essentially a two-outcome measurement $\{T_l,1-T_l\}$
and is determined by the test operator $T_l$, which satisfies the condition $T_l|\Psi\>=|\Psi\>$, so that the target state $|\Psi\rangle$ can always pass the test \cite{PallLM18,ZhuH19AdS,ZhuH19AdL}.

Suppose the test $T_l$ is chosen with probability $p_l$; then the performance of the verification procedure is determined by the \emph{verification operator}  $\Omega=\sum_{l=1} p_l T_l$.
If $\sigma$ is a quantum state that satisfies  $\bra{\Psi}\sigma \ket{\Psi} \leq 1-\epsilon$, then the  probability that $\sigma$ can pass each test on average satisfies
\begin{equation}\label{eq:PassProb}
\max_{\<\Psi|\sigma|\Psi\>\leq 1-\epsilon }\!\tr(\Omega \sigma)=1- [1-\beta(\Omega)]\epsilon=1- \nu(\Omega)\epsilon,
\end{equation}
where $\beta(\Omega)$ is the second largest eigenvalue of $\Omega$, and  $\nu(\Omega)=1-\beta(\Omega)$ is the \emph{spectral gap} from the maximum eigenvalue.  
To verify the
target state within infidelity $\epsilon$ and significance level $\delta$ (assuming $0<\epsilon,\delta<1$), the minimum number of tests required reads \cite{PallLM18,ZhuH19AdS,ZhuH19AdL}
\begin{equation}\label{eq:TestNum}
N=\biggl\lceil\frac{ \ln \delta}{\ln[1-\nu(\Omega)\epsilon]}\biggr\rceil
\leq \biggl\lceil\frac{ \ln (\delta^{-1})}{\nu(\Omega)\epsilon}\biggr\rceil
\approx \frac{ \ln (\delta^{-1})}{\nu(\Omega)\epsilon},
\end{equation}
which is  inversely proportional to the spectral gap $\nu(\Omega)$. To optimize the performance, we need to maximize the spectral gap over  the accessible measurements. The verification operator $\Omega$ is \emph{homogeneous}  if it has the following form \cite{ZhuH19AdS,ZhuH19AdL}
\begin{align}\label{eq:Homo}
\Omega=|\Psi\>\<\Psi|+\lambda(1-|\Psi\>\<\Psi|),
\end{align}
where $0\leq \lambda<1$ is a real number. Such a verification strategy is 
also useful for fidelity estimation and is thus particularly appealing.

\subsection{Subspace verification}
Next, we generalize the idea of QSV to subspace verification, which is crucial to verifying ground states of local Hamiltonians. Previously, the idea of subspace verification was employed only in some special setting \cite{FujiH17}. 
Consider a device that is supposed to produce a quantum state supported in a subspace $\caV$ within the Hilbert space $\caH$, but may actually produce something different. To this end, in each run we can perform a random test from a set of accessible tests. Each test  is determined by a test operator $T_l$ as in QSV. Let $Q$ be the projector onto the subspace $\caV$. Then the condition
$T_l|\Psi\>=|\Psi\>$ in QSV should now be replaced by $T_l Q=Q$, 
so that every state supported in $\caV$ can always pass each test.  Suppose the test $T_l$ is performed with probability $p_l$; then the performance of the verification procedure is determined by the verification operator $\Omega=\sum_{l=1} p_l T_l$, which is analogous to the counterpart in QSV.  The verification operator $\Omega$ is \emph{homogeneous}  if it has the following form [cf. \eref{eq:Homo}]
\begin{align}
\Omega=Q+\lambda(1-Q),
\end{align}
where $0\leq \lambda<1$ is a real number.

Suppose the quantum state $\sigma$ produced has fidelity at most $1-\epsilon$, which means $\tr(Q\sigma)\leq 1-\epsilon$; then the maximal probability that $\sigma$ can pass each test on average reads
\begin{equation}\label{eq:PassProbSS}
\max_{\tr(Q\sigma)\leq 1-\epsilon }\tr(\Omega \sigma)=1- [1-\beta(\Omega)]\epsilon=1- \nu(\Omega)\epsilon,
\end{equation}
where 
\begin{align}
\beta(\Omega)=\|\bar{\Omega}\|,  \quad \bar{\Omega}=\Omega-Q=(1-Q)\Omega(1-Q),
\end{align}
and  $\nu(\Omega)=1-\beta(\Omega)$ is also called the spectral gap. The  number of tests required to verify the subspace $\caV$ within infidelity $\epsilon$ and significance level $\delta$ is still  given by \eref{eq:TestNum}, although the meaning of $\nu(\Omega)$ is a bit different now. \\

\subsection{Hypergraphs}
A hypergraph $G=(V,E)$ is specified by a set of vertices $V$ and a set of edges (hyperedges) $E$, where each edge is a nonempty subset of $V$ \cite{Volo09book,ZhuH19E}. 
An edge is a loop if it contains only one vertex. Two distinct  vertices of $G$ are neighbors or adjacent  if they belong to a same edge.  The degree of a vertex $j\in V$ is the number of its neighbors and is denoted by $\deg(j)$; the degree of $G$ is the maximum vertex degree and is denoted by $\Delta(G)$. The hypergraph $G$ is connected if 
for each pair of distinct vertices $i,j$, there exist a positive integer $h$ and vertices $i_1, i_2, \ldots,i_h$ with $i_1=i$ and $i_h=j$ such that $i_k, i_{k+1}$ are adjacent for $k=1, 2, \ldots, h-1$.

Two distinct edges of $G$ are neighbors or adjacent if their intersection is nonempty. 
A subset $M$ of $E$ is a \emph{matching} of $G$ if no two edges in $M$ are adjacent. A set $\scrM$ of matchings is a \emph{matching cover} if it
covers $E$, which means $\cup_{M\in \scrM} M=E$. It should be noted that, in some literature, a matching cover means a set of matchings that covers the vertex set, which is different from our definition. An \emph{edge coloring} of $G$ is an assignment of colors to its edges such that adjacent edges have different colors. The edge coloring is trivial if no two edges are assigned with the same color. Note that every edge coloring of $G$ determines a matching cover. 
Conversely, every matching cover composed of disjoint matchings determines an edge coloring. The \emph{chromatic index} (also known as edge chromatic number) of $G$ is the minimum number of colors  required to color the edges of $G$ and is denoted by $\chi'(G)$; it is also the minimum number of matchings required to cover the edge set $E$.

A (simple) graph is a special hypergraph in which each edge contains  two vertices. According to  Vizing's theorem \cite{Vizi64,Volo09book}, the chromatic index of a graph $G(V,E)$ satisfies 
\begin{align}
\Delta(G)\leq \chi'(G)\leq \Delta(G)+1. 
\end{align}
In general,  it is computationally very demanding to find an optimal edge coloring,  but it is easy to construct a nearly optimal edge coloring  with $\Delta(G)+1\leq \chi'(G)+1$ colors~\cite{MisrG92}.\\

\subsection{Frustration-free Hamiltonians}
Since we are mainly interested in the ground states, without loss of generality, we can assume that the Hamiltonian $H$ is a sum of projectors, which share a common null vector. These projectors can be labeled by the edges (hyperedges) of a hypergraph $G=(V,E)$   \cite{Volo09book,ZhuH19E}, and $H$ can be expressed as 
\begin{align}
H=\sum_{e\in E} P_e, \label{eq:Hamiltonian}
\end{align}
where the projector $P_e$ acts (nontrivially) only on the nodes associated with the vertices contained in $e$.
Given that $H$ is frustration free by assumption, 
a state $|\Phi\>$ is a ground state iff $P_e|\Phi\>=0$ for all $e\in E$, so the ground state energy is 0. The \emph{spectral gap} of $H$ is the smallest nonzero eigenvalue and is denoted by $\gamma=\gamma(H)$ (note the distinction from the spectral gap of a verification operator). The Hamiltonian $H$ is $k$-local if each projector $P_e$ acts on at most $k$ nodes, in which case each edge of $G$ contains at most $k$ vertices.   Let $g=g(H)$ be the smallest integer $j$ such that each projector $P_e$ commutes with all other projectors $P_{e'}$ except for $j$ of them.\\

\section{\label{sec:DL}A stronger detectability lemma}
The detectability lemma  proved in \rcite{AharALV09} and improved in \rcite{AnshAV16} is a powerful tool for understanding the properties of frustration-free Hamiltonians, including the spectral gaps in particular. 
Here we shall derive a stronger version of the detectability lemma and discuss its implications. 
This  result will be very useful to deriving tighter bounds on the sample cost of our verification protocols for the ground states, which is our original motivation. We believe that it is also of independent interest to many researchers in the quantum information community.

\subsection{Improvement of the detectability lemma}

The following improvement of the detectability lemma and its corollary \lref{lem:DLnorm} are proved in \sref{sec:lem:DLproof}. 
\begin{lem}\label{lem:DL}
Let $\{P_k\}_{k=1}^q$ be a set of projectors on a given Hilbert space $\caH$, $Q_k=1-P_k$, and $H=\sum_k P_k$. Let $|\psi\>$ be any normalized ket in $\caH$, $|\varphi\>=Q_1Q_2\cdots Q_q|\psi\>$, and $\varepsilon_\varphi =\<\varphi|H|\varphi\>/\|\varphi\|^2$ with $\|\varphi\|^2=\||\varphi\>\|^2=\<\varphi|\varphi\>$  (assuming $\|\varphi\|>0$).  Then 
	\begin{align}\label{eq:DL}
	\|\varphi\|^2\!\leq\! \frac{\zeta}{\varepsilon_\varphi+\zeta}\!\leq\! \frac{s^2\tilde{g}}{\varepsilon_\varphi+s^2\tilde{g}}\!\leq\! \frac{s^2g^2}{\varepsilon_\varphi+s^2g^2}
	\!\leq\! \frac{g^2}{\varepsilon_\varphi+g^2}
	\end{align}
	with $\zeta=\max_{j} \zeta_j$ and $s=\max_{j<k}s_{jk}$, where $s_{jk}$ is the largest singular value of $P_jP_k$ that is not equal to~1 ($s_{jk}=0$ if all singular values of $P_jP_k$ are equal to 1), and
	\begin{gather}
	\zeta_j=\sum_{k|j\in \scrA_k} g_k s_{jk}^2,\quad \tilde{g}=\max_{j} \sum_{k|j\in \scrA_k} g_k,\\
	\scrA_k=\{j|j<k, P_j P_k\neq P_k P_j\}, \;\; g_k=|\scrA_k|. \label{eq:scrAkgk}
	\end{gather}
\end{lem}
Here the last upper bound in \eref{eq:DL} was derived in \rcite{AnshAV16}, while the first three bounds improve over the original result.

To appreciate the implications of \lref{lem:DL},
suppose the Hamiltonian $H$ in \lref{lem:DL} is frustration free, and $|\psi\>$  is orthogonal to the ground state space. Then $|\varphi\>=Q_1Q_2\cdots Q_q|\psi\>$ is also orthogonal to the ground state space, which means $\varepsilon_\varphi\geq \gamma= \gamma(H)$ and 
\begin{align}\label{eq:DLO}
&\|(1-P_1)(1-P_2)\cdots (1-P_q)|\psi\>\|^2\nonumber\\
&=\|Q_1Q_2\cdots Q_q|\psi\>\|^2=\|\varphi\|^2 \leq \frac{\zeta}{\gamma+\zeta}\nonumber\\
&\leq \frac{s^2\tilde{g}}{\gamma+s^2\tilde{g}}\leq \frac{s^2g^2}{\gamma+s^2g^2}
\leq \frac{g^2}{\gamma+g^2}.
\end{align}
Here  the last upper bound  was derived in \rcite{AnshAV16}. Our improvement of the detectability lemma presented in \lref{lem:DL} is crucial to deriving the first three upper bounds, which in turn are crucial to deriving \lref{lem:DLnorm}  and \thref{thm:SpectralGap} below. This improvement can sometimes significantly reduce the upper bound on the number of tests required to verify the ground state of a frustration-free Hamiltonian, as we shall see in  \sref{sec:EVGS}.

\begin{lem}\label{lem:DLnorm}
	Suppose the Hamiltonian $H$ in \lref{lem:DL} is frustration free; let $Q_0$ be the projector onto the ground state space of $H$ and 
	\begin{align}
	\!\! \bar{Q}_k=(1-Q_0)Q_k(1-Q_0),\;\; k=1,2,\ldots, q.
	\end{align}
	Then 
	\begin{align}\label{eq:DLnorm}
	\|\bar{Q}_1 \bar{Q}_2\cdots \bar{Q}_q\|^2&\leq \frac{\zeta}{\gamma+\zeta}\leq \frac{s^2\tilde{g}}{\gamma+s^2\tilde{g}}\leq \frac{s^2g^2}{\gamma+s^2g^2}\nonumber\\
	&
	\leq \frac{g^2}{\gamma+g^2}.
	\end{align}
\end{lem}
The first two upper bounds in \eref{eq:DLnorm} may depend on the order of the projectors $\bar{Q}_k$ in the product [cf.~\eref{eq:DL}], while the last two upper bounds are independent of this order.

\subsection{\label{sec:lem:DLproof}Proofs of \lsref{lem:DL} and \ref{lem:DLnorm}}
\begin{proof}[Proof of \lref{lem:DL}]
	To start with, we shall derive an upper bound for $\<\varphi|P_k|\varphi \>=\|P_k|\varphi \>\|^2$. Following \rcite{AnshAV16}, to derive an upper bound for
	\begin{align}
	\|P_k|\varphi \>\|=\|P_k(1-P_1)(1-P_2)\cdots (1-P_q)|\psi \>\|,
	\end{align}
	we can move the projector $P_k$ to the right until it is annihilated by $1-P_k$. Only those terms that do not commute with $P_k$ will contribute to the upper bound. 

For $j=1,2,\ldots, q$ let 
\begin{align}
|\varphi_{j}\>=(1-P_{j})(1-P_{j+1})\cdots (1-P_q)|\psi\>; \label{eq:varphij} 
\end{align}
then $|\varphi_1\>=|\varphi\>$.
By virtue of  \lref{lem:PQUB} below we can deduce that	
\begin{align}
\|P_k|\varphi_j\>\|&=\|P_k(1-P_j)|\varphi_{j+1}\>\|\nonumber\\
&\leq \|P_k|\varphi_{j+1}\>\|
 +s_{jk}\|P_j|\varphi_{j+1}\>\|,
\end{align}
where $s_{jk}$ is the largest singular value of $P_jP_k$ that is not equal to 1 (note that $s_{jk}=s_{kj}$). So
	\begin{align}
	\!\!\|P_k|\varphi \>\|\leq \sum_{j\in \scrA_k} s_{jk}\|P_j|\varphi_{j+1}\>\|,
	\end{align}	
	where $\scrA_k$  is defined in \eref{eq:scrAkgk}
and denotes the set of indices of the projectors $P_1, P_2, \ldots,P_{k-1}$ that do not commute with $P_k$. As a corollary, 
	\begin{align}
	\<\varphi|P_k|\varphi \>\leq g_k\sum_{j\in \scrA_k} s_{jk}^2\|P_j|\varphi_{j+1}\>\|^2, \label{eq:DLproof}
	\end{align}
	where $g_k=|\scrA_k|$ is the cardinality of $\scrA_k$.

	Next, summing over $k$ in \eref{eq:DLproof} yields
	\begin{align}
	&\<\varphi|H|\varphi\>=\sum_k \<\varphi|P_k|\varphi\>\nonumber\\
	&\leq \sum_k g_k\sum_{j\in \scrA_k} s_{jk}^2\|P_j|\varphi_{j+1}\>\|^2\nonumber\\
	&=\sum_{j=1}^{q-1} \zeta_j\|P_j|\varphi_{j+1}\>\|^2\leq \zeta\sum_{j=1}^{q-1} \|P_j|\varphi_{j+1}\>\|^2\nonumber\\
	&=\zeta\bigl[\|\varphi_q\|^2-\|\varphi_1\|^2\bigr]\leq \zeta\bigl(1-\|\varphi\|^2\bigr),
	\end{align}	
	which implies the first inequality in  \eref{eq:DL}. Here the last equality follows from the relation $|\varphi_j\>=(1-P_j)|\varphi_{j+1}\>$ 
	and the identity
	\begin{align}
	\|P_j|\varphi_{j+1}\>\|^2+\|(1-P_j)|\varphi_{j+1}\>\|^2=\|\varphi_{j+1}\|^2. 
	\end{align}	
	The rest inequalities in \eref{eq:DL} are  simple corollaries of the  facts below,
	\begin{gather}
	\!\!	\zeta_j=\sum_{k|j\in \scrA_k} g_k s_{jk}^2\leq s^2\sum_{k|j\in \scrA_k} g_k\leq s^2 \tilde{g}\quad \forall j,\\
	\tilde{g} =\max_j\sum_{k|j\in \scrA_k}  g_k\leq  g^2,\\
	0\leq s<1.
	\end{gather}
\end{proof}

The following technical lemma employed in the proof of \lref{lem:DL} is proved in \aref{app:lem:PQUBproof}
\begin{lem}\label{lem:PQUB}
	Suppose $P$ and $Q$ are two projectors on $\caH$ and $|\psi\>\in \caH$. Then 
	\begin{align}
	\|P(1-Q)|\psi\> \|\leq \|P|\psi\>\|+s \|Q|\psi\>\| , \label{eq:PQUB}
	\end{align}
	where $s$ is the largest singular value of $PQ$ that is not equal to 1 ($s=0$ if all singular values of $PQ$ are equal to 1). 
\end{lem}
\Eref{eq:PQUB} holds  even if $|\psi\>$ is not normalized.

\begin{proof}[Proof of \lref{lem:DLnorm}]
	By assumption  $Q_0$ commutes with all $P_k$ and   $Q_k$ for $1\leq k\leq q$. Let $|\psi\>$ be any normalized ket in the Hilbert space under consideration; then $(1-Q_0)|\psi\>$ is orthogonal to the ground state space. Therefore,	
	\begin{align}
	&\|\bar{Q}_1 \bar{Q}_2\cdots \bar{Q}_q |\psi\>\|^2\nonumber\\
	&=\|(1-Q_0)Q_1 Q_2\cdots Q_q (1-Q_0)|\psi\>\|^2\nonumber\\
	&\leq \|Q_1 Q_2\cdots Q_q (1-Q_0)|\psi\>\|^2
	\nonumber\\
	&\leq \frac{\zeta}{\gamma+\zeta}\|(1-Q_0)|\psi\>\|^2\leq \frac{\zeta}{\gamma+\zeta}\nonumber\\
	&\leq \frac{s^2\tilde{g}}{\gamma+s^2\tilde{g}}\leq \frac{s^2g^2}{\gamma+s^2g^2}\leq \frac{g^2}{\gamma+g^2},
	\end{align}
	which implies \eref{eq:DLnorm}. Here the second inequality follows from  \eref{eq:DLO}. 
\end{proof}

\section{\label{sec:EVGS}Efficient verification of ground states}
\subsection{Matching and coloring protocols}
Suppose the Hamiltonian in \eref{eq:Hamiltonian} has a nondegenerate ground state denoted by $|\Psi_H\>$ (the nondegeneracy assumption is included to simplify the description and is not crucial). Let $Q_e=1-P_e$; then  $P_e|\Psi_H\>=0$ and $Q_e|\Psi_H\>=|\Psi_H\>$ for all $e\in E$. To verify the ground state, we need to verify that the state is supported in the range of $Q_e$ for each $e\in G(V,E)$, which can be realized by  subspace verification. A verification protocol (operator) for an edge $e$ is referred to as a \emph{bond  verification protocol} (operator). Since $P_e$ and $Q_e$ for each edge $e$ only act  on a few nodes, it is much easier to construct bond verification protocols than protocols for  the  ground state. 
Here we  provide a general recipe for constructing verification protocols for the ground state given 
a bond verification protocol with verification operator $\Omega_e$ for each edge $e$. Note that the operator $\Omega_e$ should satisfy the conditions $Q_e\leq \Omega_e\leq 1$ and  $\Omega_e Q_e =Q_e$. Since $Q_e$ is a projector, it follows that the largest eigenvalue of $\Omega_e$ is 1 and its multiplicity is at least the rank of $Q_e$; all other eigenvalues of $\Omega_e$ belong to the interval $[0,1)$. Let 
\begin{align}
\beta_e&=\|\Omega_e-Q_e\|,& \; \nu_e&=1-\beta_e,\label{eq:betanue}\\
\beta_E&=\max_{e\in E} \beta_e, & \; \nu_E&=1-\beta_E=\min_{e\in E} \nu_e;
\end{align}
then $\nu_e$ is the spectral gap of $\Omega_e$, and $\nu_E$ is the minimum spectral gap over all bond verification operators.

In many cases of practical interest, the underlying Hamiltonian has a high symmetry (say the symmetry of a square lattice), and  it is possible to construct bond verification operators  $\Omega_e$ that are unitarily equivalent to each other.  Accordingly,  all $\beta_e$ for $e\in E$ are equal, and so are all $\nu_e$ for $e\in E$, which means $\beta_E=\beta_e$ and $\nu_E=\nu_e$.

Given a matching $M$ of $G$, we can construct a test for the ground state $|\Psi_H\>$ by performing the bond verification strategy $\Omega_e$ for each $e\in M$ independently. The resulting  test operator reads
\begin{align}
T_M=\prod_{e\in M} \Omega_e. \label{eq:TestMatch}
\end{align}
Note that all $\Omega_e$ for $e\in M$ commute with each other, so the order in the product does not matter. In addition, a state $|\Phi\>$ satisfies   $T_M|\Phi\>=|\Phi\>$ iff $P_e|\Phi\>=0$ for each $e\in M$. So  the  state  $|\Phi\>$ can pass the test $T_M$  with certainty iff it belongs to the null space of  $P_e$ for each $e\in M$. 

 \begin{figure}[b]
	\centering
	\includegraphics[width=0.48\textwidth]{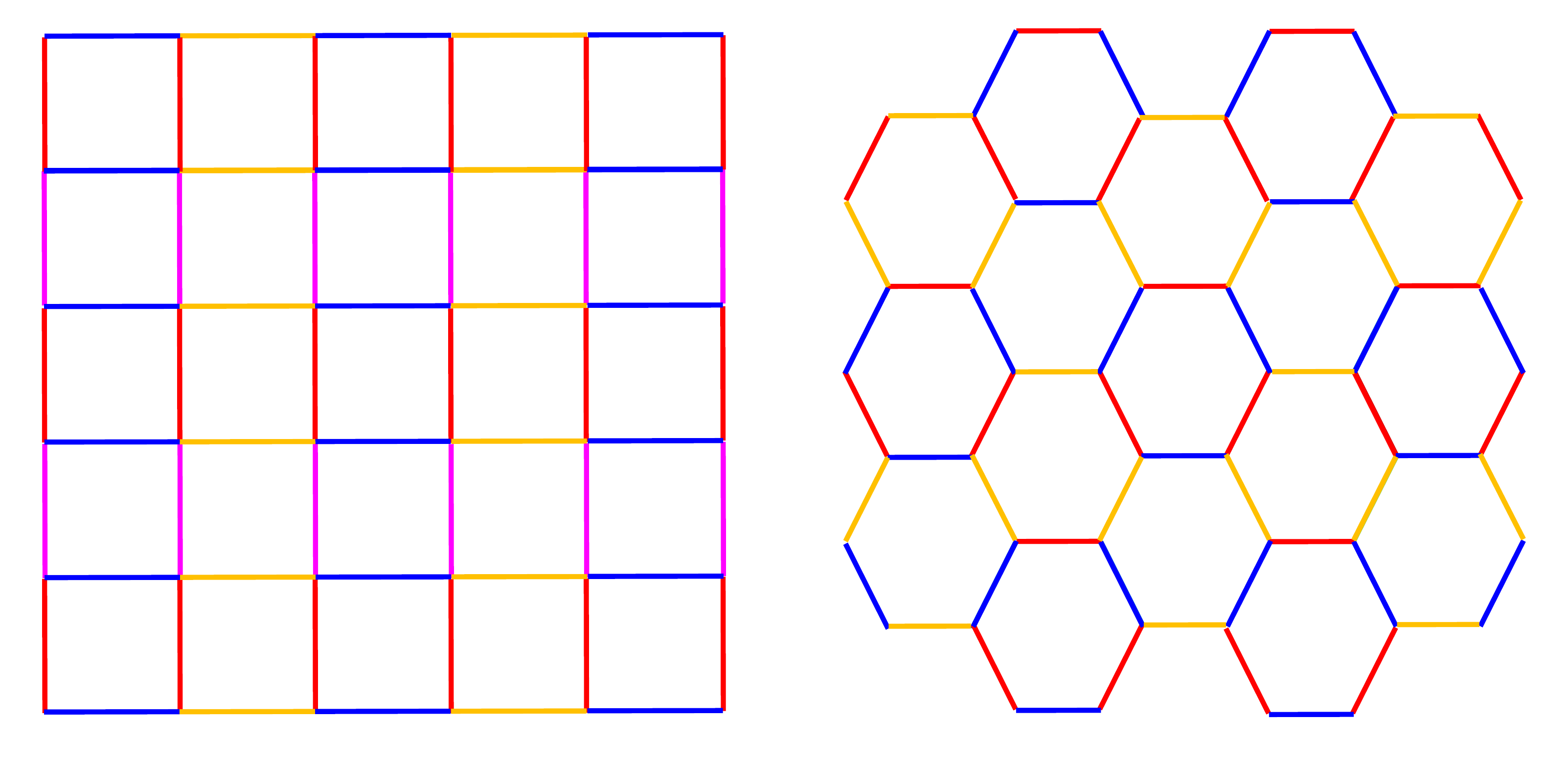}		
	\caption{\label{fig:SquareHoney}Optimal edge colorings of the square lattice and  honeycomb lattice. These optimal colorings can be used to construct efficient protocols for verifying ground states of frustration-free Hamiltonians, including 
		AKLT states. }	
\end{figure} 

Let $\scrM=\{M_l\}_{l=1}^m$ be a matching cover of $G(V,E)$ that consists of 
$m$ matchings, so that $\cup_{l=1}^m M_l =E$. For each matching $M_l$, we can construct a test $T_{M_l}$ by \eref{eq:TestMatch}. Then only the target state $|\Psi_H\>$ can pass each  test with certainty. Let $p=(p_l)_{l=1}^m$ be a probability distribution on $\scrM$, then we can construct a \emph{matching protocol} for $|\Psi_H\>$ by performing the test $T_{M_l}$ with probability $p_l$ for $l=1,2,\ldots,m$. The resulting verification operator reads 
\begin{align}
\Omega(\scrM,p)=\sum_{l=1}^m p_l T_{M_l},
\end{align}
which can be abbreviated   as $\Omega(\scrM)$ when the probability distribution $p$ is uniform, that is,
  $p_l=1/m$ for $l=1,2,\ldots, m$.

When the matchings in $\scrM$ are mutually disjoint, $\scrM$ determines an edge coloring of $G$, as illustrated in \fref{fig:SquareHoney}; the resulting protocol is  called an \emph{edge coloring protocol} or coloring protocol in short. Such a protocol has a very simple graphical description and is thus quite appealing.\\

\subsection{\label{sec:SampleC}Sample complexity}
The efficiency of the matching protocol is guaranteed by \thref{thm:SpectralGap} below, which can be proved by virtue of   
the improved detectability lemma  and quantum union bound \cite{Gao15,DonnV22}, as shown in  \sref{sec:thm:SpectralGapProof}.
\begin{thm}\label{thm:SpectralGap}
	Suppose $H$ is the frustration-free Hamiltonian in \eref{eq:Hamiltonian}. Let $\Omega(\scrM)$ be the verification operator associated with the matching cover $\scrM=\{M_l\}_{l=1}^m$ of $G(V,E)$  and bond verification operators $\{\Omega_e\}_{e\in E}$.  Then 
	\begin{align}
	\nu(\Omega(\scrM))\geq \frac{\nu_E}{m}f_m\biggl(\frac{\gamma}{s^2g^2}\biggr)\geq \frac{\nu_E\gamma}{6mg^2}, \label{eq:SpectralGapMat}
	\end{align}
	where $\nu_E=\min_{e\in E}\nu_e$ is the minimum  spectral gap of  $\Omega_e$, $s$ is defined as in \lref{lem:DL},  and 
\begin{align}\label{eq:fx}
f_m(x):=\begin{cases}
\frac{\sqrt{1+x}-1}{\sqrt{1+x}} &m=2,\\[1ex]
\frac{\sqrt{1+x}-1}{\sqrt{1+x}+1} &m\geq 3.
\end{cases}
\end{align}	
The number of tests required to verify the ground state within infidelity $\epsilon$ and significance level $\delta$ satisfies 
	\begin{align}\label{eq:TestNumUB}
N\leq \biggl\lceil\frac{m \ln (\delta^{-1})}{\nu_E\epsilon f_m\bigl(\frac{\gamma}{s^2g^2}\bigr)}
\biggr\rceil \leq 
\biggl\lceil
 \frac{6mg^2\ln (\delta^{-1})}{\nu_E\gamma\epsilon}
\biggr\rceil. 
	\end{align}
\end{thm}

The weaker bound in \eref{eq:SpectralGapMat} and that in \eref{eq:TestNumUB} can already clarify the sample complexity of the matching protocol. The stronger bounds in the two equations are slightly more complicated and rely on the stronger detectability lemma, that is, \lref{lem:DL}. On the other hand, this improvement can sometimes significantly reduce the upper bound on the number of tests (though not the scaling behavior) required to verify the ground state of a frustration-free Hamiltonian.
In the verification of the AKLT state on the honeycomb lattice for example, the first lower  bound in \eref{eq:SpectralGapMat} is about six times of the second lower bound. So  the number of tests in \eref{eq:TestNumUB} can be reduced by a factor of six thanks to the stronger bound, which can make a huge difference for practical applications.

It is instructive to analyze the lower bound for the spectral gap in \eref{eq:SpectralGapMat} when  $\gamma/(s^2g^2)\ll 1$, which holds in most cases of practical interest. When $x\ll 1$, the function $f_m(x)$  can be approximated as
\begin{align}
f_m(x)\approx\begin{cases}
\frac{x}{2} &m=2,\\[1ex]
\frac{x}{4} &m\geq 3.
\end{cases}
\end{align}
Therefore, the spectral gap can be bounded from below as follows,
\begin{align}
\!\!\!\nu(\Omega(\scrM))\geq \frac{\nu_E}{m}f_m\biggl(\frac{\gamma}{s^2g^2}\biggr)\approx
\begin{cases}
\frac{\nu_E\gamma}{2ms^2g^2} & m=2,\\[0.5ex]
\frac{\nu_E\gamma}{4ms^2g^2} & m\geq 3,
\end{cases} 
\end{align}
which is tighter than the second bound in \eref{eq:SpectralGapMat}. Accordingly, the number of tests required to verify the ground state within infidelity $\epsilon$ and significance level $\delta$ satisfies 
\begin{align}
N \lesssim
\begin{cases}
\frac{2ms^2g^2\ln (\delta^{-1})}{\nu_E\gamma\epsilon} &m=2,\\[1.5ex]
\frac{4ms^2g^2\ln (\delta^{-1})}{\nu_E\gamma\epsilon} &m\geq 3.
\end{cases} 
\end{align}

If the underlying Hamiltonian $H$ is 2-local and each projector $P_e$ acts on two nodes, then $G$ is a (simple) graph and $g\leq 2\Delta(G)-2$, where $\Delta(G)$ is the degree of $G$, so  \thref{thm:SpectralGap} implies that
\begin{align}
\nu(\Omega(\scrM))\geq  \frac{\nu_E\gamma}{24m[\Delta(G)-1]^2}. \label{eq:SpectralGapMat2local}
\end{align}		
The cardinality of the matching cover $\scrM$ is at least the chromatic index $\chi'(G)$, which satisfies $\chi'(G)\leq \Delta(G)+1$ by Vizing's theorem \cite{Vizi64,Volo09book}. If  $\scrM$ is optimal, that is, $m=|\scrM|=\chi'(G)$, then \eref{eq:SpectralGapMat2local} yields
\begin{align}
\nu(\Omega(\scrM))&\geq \frac{\nu_E\gamma}{24[\Delta(G)+1][\Delta(G)-1]^2}\nonumber\\&
\geq \frac{\nu_E\gamma}{24\Delta(G)^3}. 
\end{align}
Here $\nu_E$ and $\Delta(G)$ do not grow with the system size for most Hamiltonians of practical interest, including those defined on various lattices as illustrated in  \fref{fig:SquareHoney}. If in addition the spectral gap $\gamma$ has a universal lower bound, then the spectral gap $\nu(\Omega(\scrM))$ has a universal lower bound, so the number of tests required to verify the ground state does not grow with the system size.  Compared with previous works \cite{CramPFS10,LanyMHB17,HangKSE17,BermHSR18,GluzKEA18,TakeM18,CruzBTS22}, our approach can achieve much better scaling behaviors with respect to the system size, spectral gap $\gamma$, and  infidelity~$\epsilon$.

Since coloring protocols are special matching protocols, all results on matching protocols presented above also apply to coloring protocols. In addition, we can derive the following result tailored to coloring protocols; see \sref{sec:thm:SpectralGap2Proof} for a proof.

\begin{thm}\label{thm:SpectralGap2}
	Suppose the matching cover $\scrM$ in \thref{thm:SpectralGap} is actually an edge coloring of $G$ and  $p=(|M_1|, |M_2|, \ldots, |M_m|)/|E|$; then 
	\begin{align}
	\nu(\Omega(\scrM,p))\geq \frac{\nu_E\gamma}{|E|}.  \label{eq:SpectralGapColor}
	\end{align}
	The inequality is saturated if $\scrM$ is the trivial edge coloring with $|\scrM|=|E|$ and all bond verification operators $\Omega_e$ are homogeneous and have the same spectral gap. 
\end{thm}
If $H$ is 2-local and each projector $P_e$ acts on two nodes, then $|E|\leq n\Delta(G)/2\leq n(n-1)/2$, so \thref{thm:SpectralGap2} means
	\begin{align}
\nu(\Omega(\scrM,p))\geq \frac{2\nu_E\gamma}{n\Delta(G)}\geq \frac{2\nu_E\gamma}{n(n-1)}. \label{eq:SpectralGapColor2}
\end{align}

\subsection{\label{sec:thm:SpectralGapProof}Proof of \thref{thm:SpectralGap}}
\begin{proof}[Proof of \thref{thm:SpectralGap}]
 For $l=1,2,\ldots, m$ 	let $T_{M_l}$ be the test operators associated with the matchings $M_l$ as defined in \eref{eq:TestMatch}. Let   
	\begin{gather}
	\Pi_l:=\prod_{e\in M_l}Q_e=\prod_{e\in M_l}(1-P_e),\\
	\Omega_0(\scrM):=\frac{1}{m}\sum_{l=1}^m\Pi_l. \label{eq:Omega0}
\end{gather}
	Then $\Pi_l$ are test projectors for the ground state $|\Psi_H\>$, and $\Omega_0(\scrM)$ is a verification operator for $|\Psi_H\>$,  although in general they cannot be realized by local measurements.

Now the fact $\Omega_e\leq Q_e +\beta_E (1-Q_e)$ means
	\begin{align}
	T_{M_l}&= \prod_{e\in M_l}\Omega_e\leq \prod_{e\in M_l} [Q_e +\beta_E (1-Q_e)]\nonumber\\
	&\leq \Pi_l +\beta_E (1-\Pi_l)=\nu_E\Pi_l +\beta_E, \\
	\Omega(\scrM)&=\frac{1}{m}\sum_{l=1}^m T_{M_l}\leq \frac{1}{m}\sum_{l=1}^m(\nu_E\Pi_l +\beta_E) \nonumber\\
	&= \nu_E\Omega_0(\scrM) +\beta_E,
	\end{align}
	which in turn imply that 	 
\begin{align}
\nu(\Omega(\scrM))&\geq \nu_E\nu(\Omega_0(\scrM)). \label{eq:nunup}
\end{align}		
In conjunction with \lref{lem:SpectralGap} below we can  deduce that
	\begin{align}
	\nu(\Omega(\scrM))&\geq \frac{\nu_E}{m}f_m\biggl(\frac{\gamma}{s^2g^2}\biggr)\geq 
	\frac{\nu_E}{m}f_m\biggl(\frac{\gamma}{g^2}\biggr)\nonumber\\
	&\geq \frac{\nu_E\gamma}{6mg^2},
	\end{align}	
	which confirms  \eref{eq:SpectralGapMat}. \Eref{eq:TestNumUB} is an immediate consequence of \esref{eq:TestNum}{eq:SpectralGapMat}.
\end{proof}

Next, we prove an auxiliary lemma employed in the proof of \thref{thm:SpectralGap}.
\begin{lem}\label{lem:SpectralGap}
	Let $\Omega_0(\scrM)$ be the verification operator defined in \eref{eq:Omega0} following the premise in \thref{thm:SpectralGap}.  Then 
	\begin{align}
	\nu(\Omega_0(\scrM))&\geq\frac{1}{m}f_m\biggl(\frac{\gamma}{s^2g^2}\biggr)\geq \frac{1}{m} f_m\biggl(\frac{\gamma}{g^2}\biggr)\nonumber\\
	&\geq \frac{\gamma}{6mg^2}, \label{eq:SpectralGapId}
	\end{align}
	where $f_m(x)$ is defined in \eref{eq:fx}.
\end{lem}

\begin{proof}[Proof of \lref{lem:SpectralGap}]
	Let 
	\begin{align}
	&\bar{\Pi}_l =\Pi_l-|\Psi_H\>\<\Psi_H|,\quad l=1,2,\ldots, m,\\ 
	&\!\!\!\bar{\Omega}_0(\scrM)=\Omega_0(\scrM)-|\Psi_H\>\<\Psi_H|=\frac{1}{m}\sum_{l=1}^m \bar{\Pi}_l;
	\end{align}
	then  $\bar{\Pi}_l$ are projectors.	First, suppose the matchings in $\scrM$ are mutually disjoint, so that $\scrM$ corresponds to an edge coloring. If in addition $m\geq 3$, then 
	\begin{align}
	&\nu(\Omega_0(\scrM))=1-\|\bar{\Omega}_0(\scrM)\|\nonumber\\
	&\geq \frac{1-\|\bar{\Pi}_1\bar{\Pi}_2\cdots \bar{\Pi}_m \|}{m(1+\|\bar{\Pi}_1\bar{\Pi}_2\cdots \bar{\Pi}_m \|)}\nonumber\\
	&\geq \frac{1}{m}\frac{1-[1+(\gamma/\zeta)]^{-1/2} }{1+[1+(\gamma/\zeta)]^{-1/2}}=\frac{1}{m}\frac{[1+(\gamma/\zeta)]^{1/2} -1}{[1+(\gamma/\zeta)]^{1/2}+1} \nonumber\\
	&=   \frac{1}{m}f_m\Bigl(\frac{\gamma}{\zeta}\Bigr)\geq \frac{1}{m}f_m\biggl(\frac{\gamma}{s^2\tilde{g}}\biggr)\geq \frac{1}{m}f_m\biggl(\frac{\gamma}{s^2g^2}\biggr)\nonumber\\
	&\geq 
	\frac{1}{m}f_m\biggl(\frac{\gamma}{g^2}\biggr). \label{eq:nup}
	\end{align}	
	Here the first inequality follows from  \lref{lem:GapSumProd} below; the second inequality follows from  \lref{lem:DLnorm}, which implies that
	\begin{align}
	\|\bar{\Pi}_1\bar{\Pi}_2\cdots \bar{\Pi}_m \|^2\leq \frac{1}{1+(\gamma/\zeta)},
	\end{align}	
	where $\zeta$ is defined as in \lref{lem:DL}.
	The last three inequalities  in \eref{eq:nup} are due to the following  inequalities
	\begin{align}
	\zeta\leq s^2\tilde{g}\leq s^2 g^2\leq g^2
	\end{align}
	and the fact that the function $f_m(x)$
	is monotonically increasing in $x$ for $x\geq 0$, which is clear from the definition in \eref{eq:fx}.
	
	Meanwhile, we have  $\gamma\leq 1$ and $g\geq 1$, which means $\gamma/g^2\leq 1$. So  \eref{eq:nup} implies that
	\begin{align}
	\nu(\Omega_0(\scrM))
	&\geq \frac{1}{m}f_m\biggl(\frac{\gamma}{g^2}\biggr)
	\geq \bigl(3-2\sqrt{2}\lsp\bigr)\frac{\gamma}{mg^2}\nonumber\\
	&>\frac{\gamma}{6mg^2}, \label{eq:nup2}
	\end{align}
	which confirms \eref{eq:SpectralGapId}.
	Note that the function $f_m(x)$ is monotonically increasing  and concave in $x$ for $x\geq 0$.
	
	When $m=2$, the first inequality in \eref{eq:nup} can be improved by \lref{lem:GapSumProd} below, then \eref{eq:SpectralGapId} follows from a similar reason as presented above. 
	
	Next, we turn to the general situation in which the matchings in $\scrM=\{M_l\}_{l=1}^m$ are not necessarily disjoint. In this case, we can always construct a matching cover $\scrM'=\{M_l'\}_{l=1}^m$  composed of mutually disjoint matchings $M_l'$ that satisfy  $M_l'\subseteq M_l$ for $l=1,2,\ldots, m$. 
	Let 
	\begin{align}
	\Pi_l'&:=\prod_{e\in M_l'}Q_e, \quad 
	\Omega_0(\scrM'):=\frac{1}{m}\sum_{l=1}^m\Pi_l'. 
	\end{align}
	Then
	\begin{align}
	\Pi_l\leq \Pi_l', \quad 
	\Omega_0(\scrM)\leq \Omega_0(\scrM'),
	\end{align}
	which implies that 
	\begin{align}
	\nu(\Omega_0(\scrM))&\geq \nu(\Omega_0(\scrM'))
	\geq\frac{1}{m}f_m\biggl(\frac{\gamma}{s^2g^2}\biggr)\nonumber\\
	&
	\geq\frac{1}{m}f_m\biggl(\frac{\gamma}{g^2}\biggr)
	\geq \frac{\gamma}{6mg^2}.
	\end{align}
	This observation completes the proof of \lref{lem:SpectralGap}.
\end{proof}

The following technical lemma employed in the proof of \lref{lem:SpectralGap} is proved in \aref{app:lem:GapSumProdproof}.
\begin{lem}\label{lem:GapSumProd}
	Suppose $P_1, P_2, \ldots, P_m$ are $m$ projectors acting on the Hilbert space $\caH$. Let $O=\sum_{j=1}^m P_j/m$; then 
	\begin{equation}
	1-\|O\|\geq \frac{1-\|P_1P_2\cdots P_m \|}{m(1+\|P_1P_2\cdots P_m \|)}. \label{eq:GapSumProd}
	\end{equation}
	If $m=2$, then 
	\begin{equation}
	1-\|O\|= \frac{1-\|P_1P_2 \|}{2}. \label{eq:GapSumProd2}
	\end{equation}
\end{lem}
When $P_2=P_3=\cdots =P_m$ and $P_1$ are mutually orthogonal rank-1 projectors, we have $\|O\|=(m-1)/m$ and $\|P_1P_2\cdots P_m \|=0$, in which case the inequality in \eref{eq:GapSumProd} is saturated. So the lower bound in \eref{eq:GapSumProd} is nearly optimal without further constraints.

\subsection{\label{sec:thm:SpectralGap2Proof}Proof of  \thref{thm:SpectralGap2}}
\begin{proof}[Proof of \thref{thm:SpectralGap2}]
	Let $T_{M_l}$ be the test operator associated with the matching $M_l\in\scrM$ as defined in \eref{eq:TestMatch}. Then 
	\begin{align}
	&\Omega(\scrM,p)=\sum_{l=1}^m p_l T_{M_l} =\sum_{l=1}^m p_l\prod_{e\in M_l} \Omega_e\nonumber\\
	&\leq \sum_{l=1}^m p_l \frac{1}{|M_l|}\sum_{e\in M_l} \Omega_e=\frac{1}{|E|}\sum_{e\in E} \Omega_e\nonumber\\
	&\leq \frac{1}{|E|}\sum_{e\in E} (1-P_e+\beta_E P_e)\nonumber\\
	&=\frac{1}{|E|}\sum_{e\in E} (1-\nu_E P_e)=1-\frac{\nu_E}{|E|}H, \label{eq:SpectralGap2Proof}
	\end{align}
	which implies \eref{eq:SpectralGapColor}. Here the third equality holds because $p=(|M_1|, |M_2|, \ldots, |M_m|)/|E|$, $\cup_l M_l=E$,  and  all matchings $M_l$ in $\scrM$ are pairwise disjoint.

	If $\scrM$ is the trivial edge coloring, then each matching $M_l$ contains only one edge, so that $m=|E|$ and $p_l=1/|E|$ for $l=1,2,\ldots, |E|$. Consequently,  the first inequality in \eref{eq:SpectralGap2Proof} is saturated. If in addition all bond verification operators $\Omega_e$ are homogeneous and have the same spectral gap, then the second inequality in \eref{eq:SpectralGap2Proof} is also saturated, which means
	\begin{align}
	\Omega(\scrM,p)&=\Omega(\scrM)
	=1-\frac{\nu_E}{|E|}H,\\
	\nu(\Omega(\scrM,p))&= \nu(\Omega(\scrM))=\frac{\nu_E\gamma(H)}{|E|}.
	\end{align}
	So the inequality in \eref{eq:SpectralGapColor} is saturated in this case. 
\end{proof}

\section{\label{sec:AKLT}Efficient verification of AKLT states}
 To illustrate the power of our general recipe, here we consider  AKLT states defined on general graphs without loops; see \aref{app:Companion} and the companion paper~\cite{ChenLZ23} for more details. For any given graph $G(V,E)$, an AKLT Hamiltonian can be constructed as follows \cite{AfflKLT87,AfflKLT88,KiriK09,KoreX10}. For each vertex $j$ we assign a spin operator $\vec{S}_j=(S_{j,x}, S_{j,y}, S_{j,z})$ with spin value $S_j=\deg(j)/2$, which corresponds to a Hilbert space of dimension $2S_j+1$. Let $S_e=S_j+S_k$ for each edge $e=\{j,k\}\in E$ and  $S_E=\max_{e\in E}S_e$; then $S_E\leq \Delta(G)$.
Let $P_e$ be the projector onto the spin-$S_e$ subspace of spins $j$ and $k$; 
then the AKLT Hamiltonian can be expressed as  $H_G=\sum_{e\in E} P_e$; it is 
 frustration free and has a unique ground state \cite{KiriK09,KoreX10}, which is denoted by $|\Psi_G\>$.

\subsection{Protocols and sample complexity}
To verify the AKLT state $|\Psi_G\>$, we need  to construct suitable bond verification protocols. This is  a two-body problem for any given bond $e$, so we  focus on the two nodes $j,k$ connected by  $e$ and ignore all other nodes for the moment. 
Given any real unit vector $\vec{r}=(r_x, r_y, r_z)$ in dimension~3, let  $S_{j,\vec{r}}=\vec{S}_j\cdot \vec{r}=r_x S_{j,x} +r_y S_{j,y}+r_z S_{j,z}$ be the spin operator along direction $\vec{r}$. Then $S_{j,\vec{r}}$ has $2S_j+1$ eigenvalues, namely, $-S_j, -S_j +1, \ldots, S_j-1, S_j$. Now a bond test can be constructed as follows: both parties perform the spin measurement along direction $\vec{r}$, and the test is passed unless they both obtain the maximum eigenvalues  or both obtain the minimum eigenvalues.

Let $|+\>_{j,\vec{r}}$ ($|-\>_{j,\vec{r}}$) be the eigenstate of $S_{j,\vec{r}}$ tied to the maximum eigenvalue $S_j$  (minimum eigenvalue $-S_j$). Define  $|\pm\>_{k,\vec{r}}$ in a similar way and let 
\begin{align}\label{eq:pmer} 
\begin{aligned}
|+\>_{e,\vec{r}}&=|+\>_{j,\vec{r}}\otimes |+\>_{k,\vec{r}}, \\ |-\>_{e,\vec{r}}&=|-\>_{j,\vec{r}}\otimes |-\>_{k,\vec{r}}. 
\end{aligned}
\end{align}
Then the bond test projector  can be expressed  as
\begin{align}
R_{e,\vec{r}}:=&1-|+\>_{e,\vec{r}}\<+|-|-\>_{e,\vec{r}}\<-|, \label{eq:Rer}
\end{align}
which satisfies  $R_{e,\vec{r}}Q_e=Q_e$ as expected. The trace of this test projector reads
\begin{align}
\tr(R_{e,\vec{r}})=&(2S_j+1)(2S_k+1)-2.  \label{eq:RerTr}
\end{align}
In addition,  $R_{e,\vec{r}}=R_{e,-\vec{r}}$, so the tests associated with antipodal points on the unit sphere are identical. As shown in Sec. IV A in the companion paper \cite{ChenLZ23}, tests of the form $R_{e,\vec{r}}$ are the optimal choice based on spin measurements.

Let $\mu$ be a probability distribution on the 
unit sphere; then we can construct a bond verification protocol by performing each test $R_{e,\vec{r}}$ according to  $\mu$. The resulting bond verification operator reads
\begin{align}
\Omega_e(\mu)=\int R_{e,\vec{r}} \rmd \mu(\vec{r}).
\end{align}
According to \eref{eq:RerTr}, its trace is given by
\begin{align}
\tr[\Omega_e(\mu)]=&(2S_j+1)(2S_k+1)-2, \label{eq:OmegaemuTr}
\end{align}
which is independent of $\mu$.
Let $\mu_1$ and $\mu_2$ be two probability distributions on the unit sphere and $\mu=p\mu_1 +(1-p)\mu_2$ with $0\leq p\leq 1$. Then by definition it is straightforward to verify that
\begin{align}
\nu(\Omega_e(\mu))\geq p_1 \nu(\Omega_e(\mu_1))+p_2 \nu(\Omega_e(\mu_2)). 
\end{align}
If in addition $\mu_2$ is related to $\mu_1$ by an orthogonal transformation (rotation or reflection) on the unit sphere, then $\nu(\Omega_e(\mu_1))=\nu(\Omega_e(\mu_2))$. The two properties are summarized in the following 
 proposition, which  is very helpful to studying the spectral gap of $\Omega_e(\mu)$. 
\begin{proposition}\label{pro:numu}
	The spectral gap $\nu(\Omega_e(\mu))$ is concave in $\mu$ and invariant under orthogonal transformations  on the unit sphere. 
\end{proposition}

By \pref{pro:numu}, the spectral gap $\nu(\Omega_e(\mu))$ is maximized when 
 $\mu$ is the isotropic distribution, which leads to the \emph{isotropic protocol}; the resulting verification operator is denoted by $\Omega_e^{\textup{iso}}$. By construction $\Omega_e^{\textup{iso}}$ is invariant under orthogonal transformations, so it is block diagonal with respect to the spin subspaces associated with total spins $|S_j-S_k|,\, |S_j-S_k|+1, \ldots, S_j+S_k=S_e$,  respectively, and it can be expressed as a weighted sum of projectors onto these spin subspaces. Meanwhile, $\Omega_e^{\textup{iso}}$ satisfies the conditions $0\leq \Omega_e^{\textup{iso}}\leq 1$ and $\Omega_e^{\textup{iso}}Q_e=Q_e$, where $Q_e=1-P_e$ and $P_e$ is the projector onto the subspace associated with
 the maximum total spin $S_e$. Note that $P_e$ and $Q_e$ have ranks $2S_e+1$ and $(2S_j+1)(2S_k+1)-(2S_e+1)=4S_jS_k$, respectively. In conjunction with \eref{eq:OmegaemuTr} we can now deduce that
 $\Omega_e^{\textup{iso}}$ is homogeneous and has the form
 \begin{equation}\label{eq:Omegaiso}
\Omega_e^{\textup{iso}}=Q_e+\frac{2S_e-1}{2S_e+1}P_e,
 \end{equation}
 which means $\nu(\Omega_e^{\textup{iso}})=2/(2S_e+1)$.
The spectral gap of any other verification operator $\Omega(\mu)$ based on spin measurements satisfies $\nu(\Omega_e(\mu))\leq \nu(\Omega_e^{\textup{iso}})=2/(2S_e+1)$.

Optimal bond verification protocols can also be constructed from discrete distributions  based on (spherical) $t$-designs, which are more appealing to practical applications. 
Given a positive integer $t$,  a probability distribution on the unit sphere is a  $t$-design if the average of any polynomial of degree at most $t$ equals the average over the isotropic distribution \cite{DelsGS77,Seid01,BannB09}.  The following theorem offers a general recipe for constructing optimal bond verification protocols, which can be proved by virtue of  the theories of $t$-designs \cite{DelsGS77,Seid01,BannB09} and spin coherent states~\cite{ZhanFG90}, as shown in \sref{sec:thm:homoDesignGenProof}.
 \begin{thm}\label{thm:homoDesignGen}
 	Let $\mu$ be a probability distribution on the unit sphere and let $\mu_\sym$ be the average  of $\mu$ and its center inversion. Then the four statements below are equivalent. 
 	\begin{enumerate}
 		
 		\item $\nu(\Omega_e(\mu))=\frac{2}{2S_e+1}$.
 		
 		\item $\Omega_e(\mu)=Q_e+\frac{2S_e-1}{2S_e+1}P_e$.
 		
 		\item $\Omega_e(\mu)$ is homogeneous. 
 		
 		\item $\mu_\sym$ forms a spherical $t$-design with $t=2S_e$. 
 	\end{enumerate}
 \end{thm}
Here $\mu_\sym$  is center symmetric by construction, so  $\mu_\sym$  is a ($2S_e$)-design iff it is a ($2\lfloor S_e\rfloor$)-design.
Note that  $t$-designs for the two-dimensional sphere can be constructed using $O(t^2)$ points \cite{BondRV13,Wome17}, so optimal bond verification protocols can be constructed using $O(S_e^2)$ tests based on spin measurements. For example, the uniform distributions on the vertices of the regular tetrahedron, octahedron, cube, icosahedron, and dodecahedron are  $t$-designs with $t=2, 3, 3, 5, 5$, respectively. 
A $7$-design can be constructed from certain orbit of the rotational symmetry group of the cube \cite{ZhuKGG16,ChenLZ23}.  
A 9-design can be constructed from a suitable combination of the icosahedron and dodecahedron  \cite{HughW20,ChenLZ23}.

\begin{figure}
	\centering
	\hspace{-1ex}\includegraphics[width=0.45\textwidth]{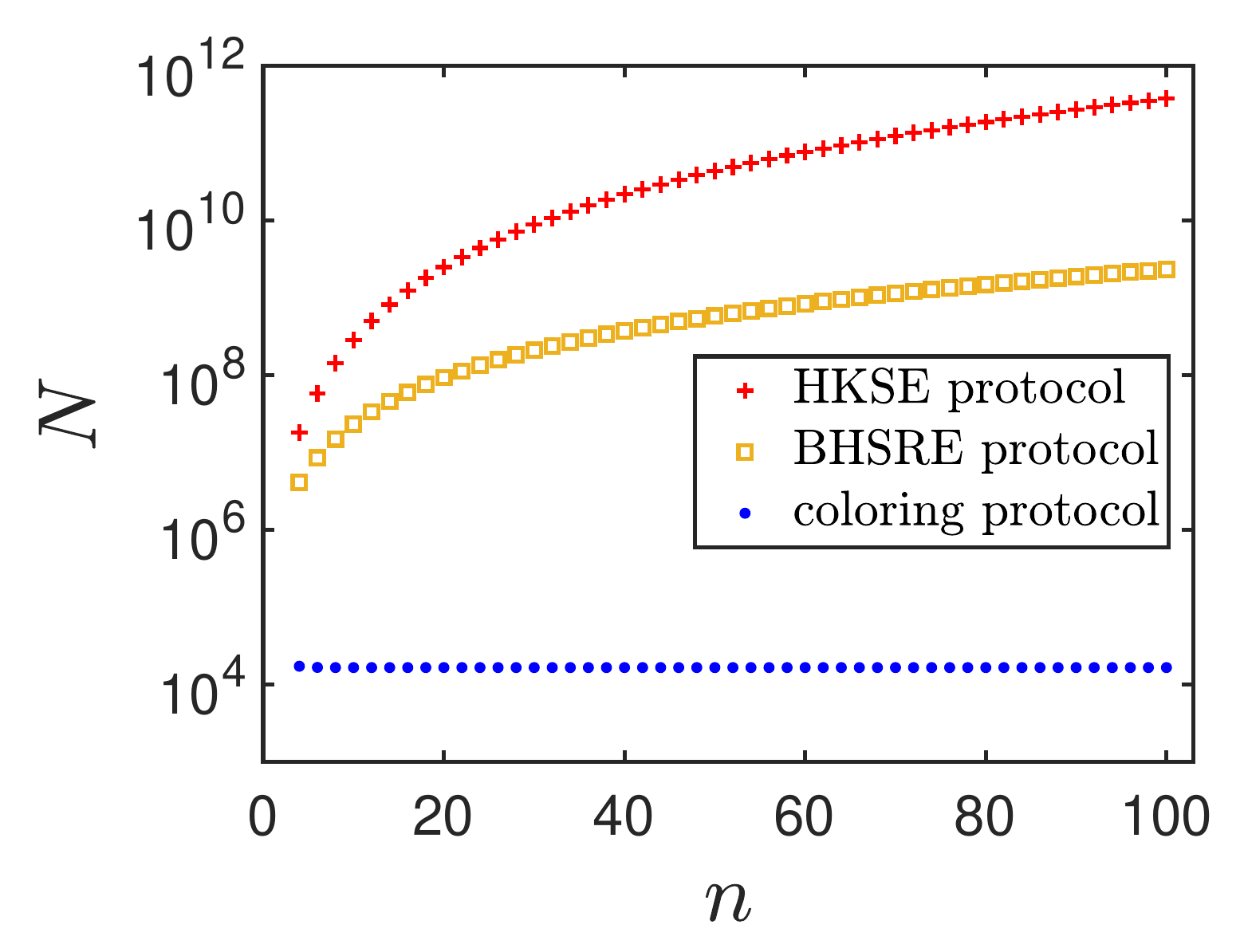}		
	\caption{\label{fig:TestNumCC}
Comparison of sample costs in the verification of	the AKLT state on the even closed chain with $n$ nodes within precision $\epsilon=\delta=0.01$. Here the  coloring protocol is a special  matching protocol based on the optimal edge coloring, and the  number of tests is determined by the first upper bound in \eref{eq:TestNumUB} with $m=g=2$, $s=1/2$, and $\nu_E=2/5$.  The HKSE protocol is proposed in \rscite{HangKSE17}, and the number of tests is determined by \eref{eq:HKSE} with $|E|=n$. The BHSRE protocol is proposed in  \rcite{BermHSR18}, and the number of tests is determined by the lower bound in \eref{eq:BHSRElb} with $\kappa=2$.
	}	
\end{figure}

For simplicity we can choose the same distribution $\mu$ for each bond verification protocol (although this is not compulsory). Let $\scrM$ be a matching cover of $G=(V,E)$ that is composed of  $m$ matchings  and let $p$ be a probability distribution on $\scrM$  (which can be omitted for the uniform distribution). Then the triple $(\mu, \scrM,p)$ specifies a verification protocol for the AKLT state $|\Psi_G\>$. Suppose $\mu$ forms a $t$-design with $t=2S_E$, $\scrM$ is an optimal matching cover (or edge coloring) with $|\scrM|=\chi'(G)$, and $p$ is uniform. 
By \thref{thm:SpectralGap} with $m=\chi'(G)\leq \Delta(G)+1$, $g=2S_E-2$, and $\nu_E=2/(2S_E+1)$,  the spectral gap of the resulting verification operator $\Omega(\mu,\scrM)$ satisfies
\begin{align}
\nu(\Omega(\mu,\scrM))&\geq \frac{\gamma}{24\Delta(G)^4},\label{eq:SpectralGapAKLT}
\end{align}
given that $S_E\leq \Delta(G)$. 
So the number of tests required to verify the AKLT state $|\Psi_G\>$ within infidelity $\epsilon$ and significance level $\delta$ satisfies 
\begin{align}
N\leq\biggl\lceil \frac{24\Delta(G)^4\ln(\delta^{-1})}{\gamma\epsilon}\biggr\rceil,
\end{align}
which leads to an upper bound that is  independent of the system size when $\gamma$ is bounded from below by a positive constant and $\Delta(G)$ is bounded from above by an integer. Notably, this is the case for various 1D and 2D lattices \cite{AfflKLT87,AfflKLT88,GarcMW13,AbduLLN20,PomaW19,PomaW20,LemmSW20,GuoPW21,WeiRA22}.

The efficiency of our approach is illustrated in \fref{fig:TestNumCC}. 
To verify the AKLT state on the closed chain with 100 nodes within precision $\epsilon=\delta=0.01$, only $1.66\times 10^4$ tests are required. For the honeycomb lattice with 100 nodes, only  $7.9\times 10^5$ tests are required. By contrast, all protocols known previously would require tens of thousands of times more tests as explained in \aref{app:Comparison}.

When the degree $\Delta(G)$ of $G$ is large (compared with $\sqrt{n}$\lsp), \thref{thm:SpectralGap2} may offer  better bounds for the spectral gap $\nu(\Omega(\mu,\scrM,p))$  and the number of tests. Suppose  $\nu_E=2/(2S_E+1)$ and $p=(|M_1|, |M_2|, \ldots, |M_m|)/|E|$; then 
\begin{align}
\!\!\nu(\Omega(\mu,\scrM,p))&\geq \frac{4\gamma}{n\Delta(G)[2\Delta(G)+1]}\geq\frac{2\gamma}{n^3},\\
N&\leq\biggl\lceil \frac{n^3\ln(\delta^{-1})}{2\gamma\epsilon}\biggr\rceil,
\end{align}
 given that $S_E\leq \Delta(G)\leq n-1$.

\subsection{\label{sec:thm:homoDesignGenProof}Proof of  \thref{thm:homoDesignGen}}
\begin{proof}[Proof of \thref{thm:homoDesignGen}]
	Let 
	\begin{align}
	O=P_e\Omega_e(\mu)P_e=\Omega_e(\mu)-Q_e;
	\end{align}
	then $O$ is a positive operator supported in the range of the  projector $P_e$,  which has rank $2S_e+1$,  and we have
	\begin{align}
	\|O\|=1-\nu(\Omega_e(\mu)). 
	\end{align}	
	In addition, 
	\begin{align}
	O&=\int \rmd\mu(\vec{r})(R_{e,\vec{r}}-Q_e)\nonumber\\
	&=\int \rmd\mu(\vec{r})(P_e-|+\>_{e,\bm{r}}\<+|-|-\>_{e,\bm{r}}\<-|),  \label{eq:O2}
	\end{align}
	which implies that  
	\begin{align}
	\tr(O)=2S_e-1.
	\end{align}	
	Suppose $\nu(\Omega_e(\mu))=2/(2S_e+1)$; then
	\begin{equation}
	\|O\|=\frac{2S_e-1}{2S_e+1}=\frac{\tr(O)}{2S_e+1}. 
	\end{equation}
	This equation implies that $O$ has rank $2S_e+1$ and  all nonzero eigenvalues are equal given that $O$ is a positive operator supported in the range of the  projector $P_e$,  which has rank $2S_e+1$. So  $O$ is necessarily  proportional to $P_e$. In conjunction with the trace formula derived above we can deduce that
	\begin{align}
	O=\frac{2S_e-1}{2S_e+1} P_e, \label{eq:MP_e}
	\end{align}	
	which confirms the implication $1\imply 2$.  The implication  $2\imply 3$	 is obvious.

	Suppose $\Omega_e(\mu)$ is homogeneous; then it has  the form 
	\begin{equation}
	\Omega_e(\mu)=Q_e+\lambda P_e,
	\end{equation}	
	which means $O=\lambda P_e$. In addition,
	\begin{align}
	\lambda=\frac{\tr(O)}{\tr(P_e)}=\frac{2S_e-1}{2S_e+1},\quad  \nu(\Omega_e(\mu))=\frac{2}{2S_e+1},
	\end{align}
	which confirms the implication $3\imply 1$. So statements 1, 2, 3 are equivalent.

	To complete the proof of \thref{thm:homoDesignGen}, it suffices to  prove the equivalence of statements 2 and 4. If statement 2 holds, then \eref{eq:MP_e} holds, which implies that
	\begin{align}
	\tr\bigl(O^2\bigr)=\frac{(2S_e-1)^2}{2S_e+1}. 
	\end{align}
	So the distribution $\mu_\sym$ forms a spherical $t$-design with $t=2S_e$ according to \lref{lem:trOmegamu2LB} below, which confirms the implication $2\imply 4$.

	If $\mu_\sym$ is a spherical $t$-design with $t=2S_e$,  then
	\begin{align}
	\tr\bigl(O^2\bigr)=\frac{(2S_e-1)^2}{2S_e+1}=\frac{1}{2S_e+1}[\tr(O)]^2
	\end{align}
	by \lref{lem:trOmegamu2LB}. Since $O$ is a positive operator supported in the range of the  projector $P_e$, which has rank $2S_e+1$, the above equation implies \eref{eq:MP_e}, and thereby confirming the implication $4\imply 2$ and completing the proof of \thref{thm:homoDesignGen}. 
\end{proof}

The following technical lemma employed in the proof of \thref{thm:homoDesignGen} is proved in \aref{app:lem:trOmegamu2LBproof}
\begin{lem}\label{lem:trOmegamu2LB}
	Let $\mu$ be a probability distribution on the unit sphere and $\Omega_e(\mu)$ the bond verification operator based on $\mu$. Then 
	\begin{align}
	\tr[\Omega_e(\mu)-Q_e]^2\geq\frac{(2S_e-1)^2}{2S_e+1},  \label{eq:trOmegamu2LB}
	\end{align}	
	and the inequality is saturated iff $\mu_\sym$ is a $t$-design with $t=2S_e$. 
\end{lem}

\section{\label{sec:sum}Summary}
We  proposed a general  recipe for verifying the ground states of frustration-free Hamiltonians based on local measurements. We also provided  rigorous bounds on the sample cost required to achieve a given precision by virtue of the spectral gap of the underlying Hamiltonian and simple graph theoretic quantities, as presented in \thsref{thm:SpectralGap} and \ref{thm:SpectralGap2}. The sample complexity achieved by our recipe is optimal with respect to the target precision as quantified by the infidelity and significance level. 
 When the Hamiltonian is local and gapped in the thermodynamic limit, the sample complexity is independent of the system size and is thus optimal with respect to the system size as well. Nevertheless, we believe that the stronger detectability lemma (\lref{lem:DL}) we proved can be further improved by a constant factor. Accordingly, the constant in \thref{thm:SpectralGap} can be further improved.
In any case, our approach can achieve much better scaling behaviors with respect to the system size, spectral gap, and  infidelity compared with alternative approaches known before.

 To demonstrate the power of this recipe,  we constructed concrete protocols for verifying AKLT states defined on arbitrary graphs based on local spin measurements, which are dramatically more efficient than previous protocols. For AKLT states defined on many lattices, including the 1D chain and honeycomb lattice, the sample cost does not increase with the system size. 
Our work reveals an intimate connection between the quantum verification problem and many-body physics. The protocols we constructed are useful not only to addressing various tasks in quantum information processing, but also to  studying  many-body physics.

The verification strategy considered in this work can make a meaningful conclusion only if all tests are passed. In practice, it is unrealistic to prepare the perfect target state, and even  states with a high fidelity may fail to pass  some tests when the total test number $N$ is large. In addition, imperfection in the measurement devices may also cause some failures.
To construct a robust verification protocol, it is necessary to allow certain failure rate. Fortunately, this extension will only incur 
a constant overhead~\cite{LiZH23}. For a typical choice of the allowed failure rate, the overhead is about ten times.  So all of our conclusions are still applicable after minor modification even if robustness is taken into account. In addition, by virtue of the recipe proposed in \rscite{ZhuH19AdS,ZhuH19AdL,LiZH23}, our protocols can be generalized to the adversarial scenario in which the preparation device cannot be trusted; moreover, the sample overhead is negligible if robustness is not a concern or if homogeneous strategies can be constructed. In general, it is still an open problem to construct robust and efficient verification protocols in the adversarial scenario, which deserves further study.

Our verification protocols can only extract information about the fidelity with the target state. Nevertheless, this key information is very useful in many applications in quantum information processing. One of the main goals in the active research area of quantum characterization, verification, and validation (QCVV) \cite{EiseHWR20,CarrEKK21,KlieR21,YuSG22,MorrSGD22} is to extract such key information as efficiently as possible given the limited resources available, which is the best we can do.  Even if this key information is not enough in certain situations, it is still valuable as a first-step diagnosis before taking more sophisticated methods. To extract more information necessarily means more resource overhead. Actually, other characterization methods, such as quantum tomography, usually require substantially (often exponentially) more resources.

\acknowledgments
H. Zhu is grateful to Zheng Yan and Penghui Yao for stimulating discussions. The authors are grateful to two anonymous reviewers for many constructive comments and suggestions. This work is  supported by   the National Natural Science Foundation of China (Grants No.~92165109 and No.~11875110),  National Key Research and Development Program of China (Grant No. 2022YFA1404204), and Shanghai Municipal Science and Technology Major Project (Grant No.~2019SHZDZX01).


\appendix
\numberwithin{equation}{section}
\renewcommand{\theequation}{\thesection\arabic{equation}}

\section*{Appendix}

In this appendix  we first prove three auxiliary lemmas employed in the main text, namely \lsref{lem:PQUB}, \ref{lem:GapSumProd}, and~\ref{lem:trOmegamu2LB}. Then we briefly discuss the connections with and distinctions from the companion paper \cite{ChenLZ23}. Finally, we compare our verification protocols with previous protocols known in the literature.

\section{\label{app:lem:PQUBproof}Proof of \lref{lem:PQUB}}

\begin{proof}[Proof of \lref{lem:PQUB}]
	When $\caH$ has dimension 0 or~1, the inequality in \eref{eq:PQUB} 	is trivial. In addition, it is easy to verify this  inequality 
	when one of the following four conditions holds, 
	\begin{enumerate}
		\item $P=0$ or $Q=0$;
		\item $P=1$ or $Q=1$;
		\item $P=Q$;
		\item $PQ=0$.
	\end{enumerate}
	So we can exclude these cases in the following discussion. Furthermore, we can assume that $|\psi\>$ is a normalized ket without loss of generality.
	
	Suppose $\caH$ has dimension 2, which is the simplest nontrivial case. In view of the above analysis, we can assume that  $P$ and $Q$ are distinct rank-1 projectors that are not orthogonal to each other. Then $P$ and $Q$ correspond to two distinct pure states, denoted by $|\alpha\>$ and $|\beta\>$ hence forth. Let $|\beta^\bot\>$ be the (normalized) ket that is orthogonal to $|\beta\>$. Let $\vec{a},\vec{b},\vec{c}$ be the Bloch vectors of  $|\alpha\>,|\beta^\bot\>,|\psi\>$, respectively. Let $\theta$ be the angle between $\vec{a}$ and $\vec{b}$ and let $\phi$ be the angle between $\vec{b}$ and $\vec{c}$, where $0< \theta<\pi$ and $0\leq \phi\leq \pi$, so that  $0< \theta+\phi<2\pi$.
	Then we have 
	\begin{gather}
	s=\|P Q\|=\sin\frac{\theta}{2},\\
\|P(1-Q)|\psi\> \|=|\<\alpha|\beta^{\bot}\>\<\beta^\bot|\psi\>|=\cos\frac{\theta}{2}\cos\frac{\phi}{2},\quad \\ 
	\|Q|\psi\> \|=|\<\beta|\psi\>|=\sin\frac{\phi}{2}.
	\end{gather}

According to the definitions of the vectors $\vec{a}, \vec{c}$ and angles $\theta,\phi$ introduced above it is easy to verify that
	\begin{align}
	\vec{a}\cdot \vec{c}\geq \cos(\theta +\phi),
	\end{align}
	which implies that 
	\begin{align}
&\|P|\psi\> \|=|\<\alpha|\psi\>|=\sqrt{\frac{1+\vec{a}\cdot \vec{c}}{2}}\nonumber\\
	&	
	\geq \sqrt{\frac{1+\cos(\theta +\phi)}{2}}=\biggl|\cos\frac{\theta +\phi}{2}\biggr|\geq \cos\frac{\theta +\phi}{2}.
	\end{align}
	Therefore,
	\begin{align}
	&\|P|\psi\>\|+s \|Q|\psi\>\|\geq \cos\frac{\theta +\phi}{2}+\sin\frac{\theta}{2}\sin\frac{\phi}{2}\nonumber\\
	&=\cos\frac{\theta}{2}\cos\frac{\phi}{2}=\|P(1-Q)|\psi\> \|,
	\end{align}
	which confirms the inequality in \eref{eq:PQUB}. 
	
	Now we are ready to consider the most general situation. Denote by $r_1$ and $r_2$ the ranks of $P$ and $Q$, respectively,  and let $r=\min\{r_1, r_2\}$. Then $P$  and $Q$ have spectral decompositions
	\begin{align}
	P=\sum_{j=1}^{r_1} |\alpha_j\>\<\alpha_j|, \quad Q=\sum_{k=1}^{r_2} |\beta_k\>\<\beta_k|,
	\end{align}
	which satisfy 
	\begin{align}
	\<\alpha_j|\beta_k\>=\tilde{s}_j\delta_{jk}, \quad 0\leq \tilde{s}_j\leq 1 \label{eq:sjtilde}
	\end{align}
	for $j=1,2,\ldots, r_1$ and $k=1,2,\ldots, r_2$.
	Without loss of generality, we can assume that $\tilde{s}_j=0$ if  $j> r$.

	Given $j=1,2,\ldots, r$, let $\caH_j$ be the subspace spanned by $|\alpha_j\>$ and $|\beta_j\>$ and let  
	\begin{align}
	\caH_0=(\caH_1+\caH_2+\cdots+\caH_r)^\bot. 
	\end{align}
	Then the subspaces $\caH_0, \caH_1, \ldots, \caH_r$ are mutually orthogonal. For $j=0,1,\ldots, r$, let $\Pi_j$ be the orthogonal projector onto $\caH_j$  and let  
	\begin{align}
	P_j=\Pi_j P\Pi_j,\quad  Q_j=\Pi_jQ\Pi_j,\quad
	|\psi_j\>=\Pi_j|\psi\>.
	\end{align}
	Then  $\|P_jQ_j\|=\tilde{s}_j$, where $\tilde{s}_j$ for $j=1,2,\ldots, r$ are introduced in \eref{eq:sjtilde}, while $\tilde{s}_0=0$ (note that $P_0Q_0=0$). 
	
	By virtue of the above analysis on the qubit and several special cases we can deduce that
	\begin{align}
	&\|P(1-Q)|\psi_j\> \|=\|P_j(1-Q_j)|\psi_j\> \|\nonumber\\
	&\leq \|P_j|\psi_j\>\|+s_j \|Q_j|\psi_j\>\|\nonumber\\
	&=\|P|\psi_j\>\|+s_j \|Q|\psi_j\>\|\leq \|P|\psi_j\>\|+s \|Q|\psi_j\>\|
	\end{align}
for $j=0,1,\ldots, r$,	where
	\begin{align}
	s_j=\begin{cases}
	\tilde{s}_j &\tilde{s}_j<1,\\
	0 & \tilde{s}_j=1.
	\end{cases}
	\end{align} 
	Note that $s=\max_{j=0}^r s_j$. 
	Therefore,
	\begin{align}
	&\|P(1-Q)|\psi\> \|^2 =\sum_{j=0}^r \|P(1-Q)|\psi_j\> \|^2\nonumber\\
	&\leq \sum_{j=0}^r (\|P|\psi_j\>\|+s \|Q|\psi_j\>\|)^2 \nonumber \\
	&=\sum_{j=0}^r \bigl(\|P|\psi_j\>\|^2 +s^2 \|Q|\psi_j\>\|^2\bigr) \nonumber\\
	&\quad +2s\sum_{j=0}^r \|P|\psi_j\>\| \|Q|\psi_j\>\| \nonumber \\
	&\leq \|P|\psi\>\|^2 +s^2 \|Q|\psi\>\|^2+2s\|P|\psi\>\| \|Q|\psi\>\|\nonumber\\
	&=(\|P|\psi\>\| +s\|Q|\psi\>\|)^2,
	\end{align}
	which implies \eref{eq:PQUB} and completes the proof of \lref{lem:PQUB}.
\end{proof}

\section{\label{app:lem:GapSumProdproof}Proof of \lref{lem:GapSumProd}}
\begin{proof}[Proof of \lref{lem:GapSumProd}]
	\Eref{eq:GapSumProd2} is a simple corollary of  Lemma~1 in  \rcite{LiHSS21}, so it remains to prove \eref{eq:GapSumProd}. 	
	Let $|\psi\>\in \caH$ be any normalized ket, and 
	\begin{align}
	x&=\|P_1P_2 \cdots P_m |\psi\>\|,\\ 
	y&=\sum_{j=1}^m \langle \psi|1-P_j|\psi\rangle=m-m\<\psi|O|\psi\>. 
	\end{align} 
	According to Theorem~1.3 in \rcite{DonnV22}, which is a generalization of the quantum union bound \cite{Gao15},	we have 
	\begin{align}
	x+\sqrt{1-x^2}\sqrt{y}\geq 1, 
	\end{align}
	which implies that
	\begin{align}
	y\geq \frac{(1-x)^2}{1-x^2}=\frac{1-x}{1+x}. 
	\end{align}
	Now choose $|\psi\>$ as an eigenvector of $O$ associated with the largest eigenvalue, that is,  $\<\psi|O|\psi\>=\|O\|$. Then the above equation implies that 
	\begin{align}
	m-m\|O\|&\geq \frac{1-\|P_1P_2\cdots P_m |\psi\> \|}{1+\|P_1P_2\cdots P_m |\psi\>\|}\nonumber\\
	&\geq \frac{1-\|P_1P_2\cdots P_m  \|}{1+\|P_1P_2\cdots P_m \|},
	\end{align}
	which in turn implies \eref{eq:GapSumProd}. 
\end{proof}

\section{\label{app:lem:trOmegamu2LBproof}Proof of \lref{lem:trOmegamu2LB}}

\begin{proof}[Proof of \lref{lem:trOmegamu2LB}]
As in the proof of  \thref{thm:homoDesignGen}, let
\begin{align}
O=\Omega_e(\mu)-Q_e=\int \rmd\mu(\vec{r})(R_{e,\vec{r}}-Q_e),
\end{align}	
where $R_{e,\vec{r}}$ is  defined in \eref{eq:Rer}.
	The inequality in \eref{eq:trOmegamu2LB} follows from the equality $\tr O=2S_e-1$ and the fact that $O$ is a positive operator supported in the range of the projector $P_e$, which has rank $2S_e+1$. 
	
	Now, by virtue of  \lref{lem:ErEsOverlap} below we can deduce that	
	\begin{align}
\tr (O^2) =2S_e-3+2^{2-2S_e}\sum_{j=0}^{\lfloor S_e\rfloor}\binom{2S_e}{2j}F_{2j}(\mu),
\end{align}	
where
	\begin{align}
	F_t(\mu):=\iint \rmd\mu(\vec{r})\rmd\mu(\vec{s}) (\bm{r}\cdot\bm{s})^t
	\end{align}
	is the $t$th frame potential of the distribution $\mu$. Note that $F_0(\mu)=1$ irrespective of the distribution $\mu$.  When $t$ is an even positive integer,  the frame potential $F_t$ satisfies the inequality \cite{Seid01,BannB09}
	\begin{align}
	F_t(\mu)=F_t(\mu_\sym)\geq \frac{1}{t+1},
	\end{align}
	which is saturated if $\mu_\sym$ forms a spherical $t$-design. Therefore,
	\begin{align}
	&\tr(O^2)\geq 2S_e-3+2^{2-2S_e}\sum_{j=0}^{\lfloor S_e\rfloor}\binom{2S_e}{2j}\frac{1}{2j+1}\nonumber\\
	&=2S_e-3+\frac{2^{2-2S_e}}{2S_e+1}\sum_{j=0}^{\lfloor S_e\rfloor}\binom{2S_e+1}{2j+1}\nonumber\\
	&=2S_e-3+\frac{4}{2S_e+1}=\frac{(2S_e-1)^2}{2S_e+1}, \label{eq:trOmegamu2LBProof}
	\end{align}
	which reproduces the inequality in \eref{eq:trOmegamu2LB}. Here the second equality follows from the identity below,
\begin{align}
\sum_{j=0}^{\lfloor S_e\rfloor}\binom{2S_e+1}{2j+1}=2^{2S_e}. 
\end{align}

	If the distribution $\mu_\sym$ forms a spherical $t$-design with $t=2S_e$, then 
	\begin{align}
	\!\! F_{2j}(\mu)=F_{2j}(\mu_\sym)=\frac{1}{2j+1}, \;\; j=0, 1, \ldots, \lfloor S_e\rfloor,   \label{eq:f2jmudesign}
	\end{align}
	so the inequality in \eref{eq:trOmegamu2LBProof} is saturated, and  the inequality in \eref{eq:trOmegamu2LB} is saturated accordingly.

	Conversely, if the inequality in \eref{eq:trOmegamu2LB} is saturated, then the inequality in \eref{eq:trOmegamu2LBProof} is saturated, so  \eref{eq:f2jmudesign} holds.  Therefore, $\mu_\sym$ forms a spherical $t$-design with $t=2S_e$ given that $\mu_\sym$ is center symmetric by construction. Note that   $\mu_\sym$  is a ($2S_e$)-design iff it is a ($2\lfloor S_e\rfloor$)-design.
\end{proof}

In the rest of this appendix we prove an auxiliary lemma employed in the proof of \lref{lem:trOmegamu2LB}.
\begin{lem}\label{lem:ErEsOverlap}
	Let  $R_{e,\vec{r}}$ and $R_{e,\vec{s}}$ be test projectors defined according to   \eref{eq:Rer}. Then 
	\begin{align}
	&\tr[(R_{e,\vec{r}}-Q_e)(R_{e,\vec{s}}-Q_e)]
	=\tr(\tilde{R}_{e,\vec{r}} \tilde{R}_{e,\vec{s}})\nonumber\\	
	&=2S_e-3+2\biggl(\frac{1+\vec{r}\cdot\vec{s}}{2}\biggr)^{2S_e}+2\biggl(\frac{1-\vec{r}\cdot\vec{s}}{2}\biggr)^{2S_e}\nonumber\\
	&=2S_e-3+2^{2-2S_e}\sum_{j=0}^{\lfloor S_e\rfloor}\binom{2S_e}{2j}(\vec{r}\cdot \vec{s})^{2j},  \label{eq:ErEsOverlap}
	\end{align}
	where $\tilde{R}_{e,\vec{r}}=P_e-|+\>_{e,\bm{r}}\<+|-|-\>_{e,\bm{r}}\<-|$.
\end{lem}
With  \lref{lem:ErEsOverlap} it is easy to compute the trace $\tr(R_{e,\vec{r}} R_{e,\vec{s}})$ since 
\begin{align}
\tr(R_{e,\vec{r}} R_{e,\vec{s}})&=\tr[(R_{e,\vec{r}}-Q_e)(R_{e,\vec{s}}-Q_e)]\nonumber\\
&\quad +\tr(Q_e). 
\end{align}

\begin{proof}[Proof of \lref{lem:ErEsOverlap}]
	By definitions in \esref{eq:pmer}{eq:Rer}, the kets  $|\pm\>_{e,\bm{r}}$ for any unit vector $\vec{r}$ in dimension 3 belong to the support of $P_e$, so $R_{e,\vec{r}}$  commutes with $P_e$ and $Q_e$. In addition,
	\begin{align}
	&R_{e,\vec{r}}-Q_e=P_eR_{e,\vec{r}}P_e =P_e R_{e,\vec{r}}=R_{e,\vec{r}}P_e\nonumber\\
	&=P_e-|+\>_{e,\bm{r}}\<+|-|-\>_{e,\bm{r}}\<-|=\tilde{R}_{e,\vec{r}}. \label{eq:ErEsOverlapProof1}
	\end{align}
	Similar conclusions also hold if $\vec{r}$ is replaced by~$\vec{s}$.

	According to the theory of spin (or atomic) coherent states  (see Sec. III D in \rcite{ZhanFG90}), we have 
	\begin{equation}
	\begin{aligned}
	&|{}_{j,\bm{r}}\<+|+\>_{j,\bm{s}}|^2 =|{}_{j,\bm{r}}\<-|-\>_{j,\bm{s}}|^2=\biggl(\frac{1+\vec{r}\cdot\vec{s}}{2}\biggr)^{2S_j}, \\
	&|{}_{j,\bm{r}}\<+|-\>_{j,\bm{s}}|^2 =|{}_{j,\bm{r}}\<-|+\>_{j,\bm{s}}|^2=\biggl(\frac{1-\vec{r}\cdot\vec{s}}{2}\biggr)^{2S_j}, 
	\end{aligned}
	\end{equation}
	which implies that
	\begin{equation}
	\begin{aligned}
	&|{}_{e,\bm{r}}\<+|+\>_{e,\bm{s}}|^2 =|{}_{e,\bm{r}}\<-|-\>_{e,\bm{s}}|^2=\biggl(\frac{1+\vec{r}\cdot\vec{s}}{2}\biggr)^{2S_e}, \\
	&|{}_{e,\bm{r}}\<+|-\>_{e,\bm{s}}|^2 =|{}_{e,\bm{r}}\<-|+\>_{e,\bm{s}}|^2=\biggl(\frac{1-\vec{r}\cdot\vec{s}}{2}\biggr)^{2S_e}. \label{eq:ErEsOverlapProof2}
	\end{aligned}
	\end{equation}

	\Esref{eq:ErEsOverlapProof1}{eq:ErEsOverlapProof2} together imply that
	\begin{align}
	&\tr[(R_{e,\vec{r}}-Q_e)(R_{e,\vec{s}}-Q_e)]
	=\tr(\tilde{R}_{e,\vec{r}} \tilde{R}_{e,\vec{s}})\nonumber\\
	&=\tr(P_e)-4+2|{}_{e,\bm{r}}\<+|+\>_{e,\bm{s}}|^2+2|{}_{e,\bm{r}}\<+|-\>_{e,\bm{s}}|^2\nonumber\\
	&=2S_e-3
	+2\biggl(\frac{1+\vec{r}\cdot\vec{s}}{2}\biggr)^{2S_e}+2\biggl(\frac{1-\vec{r}\cdot\vec{s}}{2}\biggr)^{2S_e}\nonumber\\
	&=2S_e-3+2^{2-2S_e}\sum_{j=0}^{\lfloor S_e\rfloor}\binom{2S_e}{2j}(\vec{r}\cdot \vec{s})^{2j},
	\end{align}	
	which confirms \eref{eq:ErEsOverlap} and completes the proof of \lref{lem:ErEsOverlap}. 
\end{proof}

\section{\label{app:Companion}Connections with and distinctions from the companion paper \cite{ChenLZ23}}
In this paper, we  proposed a general  recipe for verifying the ground states of frustration-free Hamiltonians based on local measurements and provided  rigorous bounds on the sample cost required to achieve a given precision as formulated in  \thsref{thm:SpectralGap} and \ref{thm:SpectralGap2}. To illustrate the power of this general recipe, we  constructed  efficient  protocols for verifying AKLT states based on local spin measurements.

In the companion paper \cite{ChenLZ23}, we discussed in more details the verification of AKLT states following the general recipe proposed here and rigorous bounds on the sample cost  as presented in  \thsref{thm:SpectralGap} and \ref{thm:SpectralGap2}. To be specific, we proved that the tests of the form in \eref{eq:Rer} are optimal among tests based on local spin measurements. In addition, the main properties of these tests and corresponding bond verification operators can actually be understood from the perspective of one spin with a suitable spin value instead of two spins, which is very helpful to technical analysis and numerical calculation. Notably, Theorem~3 in \rcite{ChenLZ23}, which is the counterpart of \thref{thm:homoDesignGen} in this paper, is formulated in this perspective. Based on this observation we constructed many  concrete bond verification  protocols based on platonic solids and other special distributions on the unit sphere, which are optimal or nearly optimal. Then  we discussed the optimization of the spectral gaps based on the optimization of matching covers and test probabilities.  Furthermore, we constructed a number of concrete verification protocols for AKLT states on open and closed chains and compared their efficiencies based on theoretical analysis and numerical calculation. Finally, we considered AKLT states on arbitrary graphs up to five vertices.

\section{\label{app:Comparison}Comparison with previous works}
In this section we compare our verification protocols for the ground states of frustration-free local Hamiltonians  with previous works \cite{TakeM18,CramPFS10,LanyMHB17,HangKSE17,BermHSR18,CruzBTS22,GluzKEA18}. It should be noted that the protocols proposed in \rscite{TakeM18,CramPFS10,LanyMHB17,BermHSR18,GluzKEA18} are also applicable to certain local Hamiltonians that are not frustration free. The property of being frustration-free is crucial to 
attaining the high efficiency  achieved in this work.

According to the following analysis, previous   protocols  require at least $O(n^2 (\ln\delta^{-1})/(\gamma^2\epsilon^2))$ tests  to verify an $n$-qubit target state within infidelity $\epsilon$ and significance level $\delta$, where $\gamma$ is the spectral gap of the underlying Hamiltonian. In sharp  contrast, our protocols require only $O((\ln\delta^{-1})/(\gamma\epsilon))$ tests, which is substantially more efficient than previous protocols. Notably, only our protocols  can verify the target state with a sample cost that is independent of the system size when  the Hamiltonian is gapped in the thermodynamic limit). Moreover, for large and intermediate quantum systems of practical interest (which are beyond classical simulation and are required to demonstrate quantum advantage), only our protocols can achieve the verification task with 
reasonable sample cost acceptable
in experiments or practical applications. When $\epsilon=\delta=0.01$ and $n$ is around 100 for example, other  protocols  would require tens of thousands of times more samples to achieve the same precision.

\subsection{Comparison with \rcite{TakeM18}}
In \rcite{TakeM18}, Takeuchi and  Morimae (TM) introduced a protocol for verifying the ground states of local Hamiltonians. To verify an $n$-qubit target state within infidelity $\epsilon=1/n$ and significance level $\delta=1/n$, the number of tests required is given by $k+m$ with
\begin{align}
k&\geq 32R^2 n^5, \quad
m\geq 2n^5k^2\log 2\geq  2^{11} n^{15}R^4\log2,
\end{align}
where $R=\poly(n)/\gamma$, and $\gamma=\gamma(H)$ is the spectral gap of the underlying Hamiltonian $H$. This number
is approximately proportional to the fourth power of the inverse spectral gap. The scaling behaviors with $\epsilon$ and $\delta$ are not clear because the choices of these parameters in  \rcite{TakeM18} are coupled with the number $n$ of qubits. In any case, the number of required tests is astronomical for any verification task of practical interest. When $n=100$ for example, this number is at least $10^{33}R^4$, which is billions of times more than what is required in our protocols and is  too prohibitive for practical applications.

Incidentally, the TM protocol can be applied to  the adversarial scenario in which the preparation device is not trusted. 
By virtue of the recipe proposed in \rscite{ZhuH19AdS,ZhuH19AdL}, our protocols can also be generalized to the adversarial scenario with negligible overhead in the sample cost. In this scenario, our protocols are still dramatically more efficient than the TM protocol.

\subsection{Comparison with \rscite{CramPFS10,LanyMHB17,HangKSE17,BermHSR18,CruzBTS22}}
In  \rcite{CramPFS10}, Cramer et al. introduced an approach for verifying quantum states that can be  approximated by matrix product states (MPS). It is much more efficient than quantum state tomography, but the paper did not give very specific sample cost (see the followup \rcite{LanyMHB17} for a bit more detail). In addition, this approach
only applies to one-dimensional systems, which is a bit restricted.

In \rcite{HangKSE17}, Hangleiter, Kliesch, Schwarz, and Eisert (HKSE) extended the approach introduced in \rcite{CramPFS10} and proposed  a general protocol for verifying the ground states of local Hamiltonians, assuming that each local projector can be measured directly.
To verify the  target state within infidelity $\epsilon$ and significance level $\delta$, the number of tests required is given by 
\begin{align}
\frac{|E|^3}{2\gamma^2 \epsilon^2}\ln\biggl[-\frac{|E|+1}{\ln(1-\delta)}\biggr]\approx \frac{|E|^3}{2\gamma^2 \epsilon^2}\ln\frac{|E|}{\delta}, \label{eq:HKSE}
\end{align}
where $|E|$ is the number of edges of the graph $G(V,E)$ encoding the action of the Hamiltonian, that is, the number of local projectors. Here the approximation is applicable  when $|E|\gg 1$ and $\delta\ll 1$, which hold for most cases of practical interest. If the Hamiltonian is 2-local and is defined on a lattice with $n$ nodes and coordination number $z$, then $|E|=zn/2$, so the above equation reduces to 
\begin{align}
\frac{z^3n^3}{16\gamma^2 \epsilon^2}\ln\frac{zn}{2\delta}.
\end{align}
This number is (approximately) proportional to
$n^3$, $\gamma^{-2}$, and $\epsilon^{-2}$. The sample complexity is much better than the  TM protocol \cite{TakeM18} discussed above. However, this number is still too prohibitive for large and intermediate quantum systems of practical interest, which consist of more than 100 qubits or qudits.

In \rcite{BermHSR18}, Bermejo-Vega,  Hangleiter, Schwarz, Raussendorf, and Eisert (BHSRE) improved the HKSE protocol and reduced the sample cost to [see Eq.~(E10) in the paper]
\begin{align}
\frac{n^2\alpha^2\kappa^2}{2\gamma^2 \epsilon^2}\ln\frac{\kappa+1}{\delta}\geq \frac{n^2}{2\gamma^2 \epsilon^2}\ln\frac{\kappa+1}{\delta}, \label{eq:BHSRElb}
\end{align}
where $\alpha,\kappa$ depend on certain decomposition of the underlying Hamiltonian and satisfy the conditions $\kappa\geq 2$ and  $\alpha\kappa\geq 1$. Here the scaling behavior in $n$ is better than the HKSE protocol, but the scaling behaviors in $\gamma$ and $\epsilon$ remain the same. 

Following the idea of the HKSE protocol, Ref.~\cite{CruzBTS22} introduced a protocol for verifying a family of tensor network states. The sample cost is proportional to $|E|^2/(\gamma^2 \epsilon^2)$, which is comparable to the BHSRE protocol. However, the proof in Ref.~\cite{CruzBTS22} relies on the Gaussian approximation, which is not very rigorous. Notably, Eq.~(E2) in Ref.~\cite{CruzBTS22} is applicable only under 
suitable restrictions on the parameters, which were not clarified. Additional analysis is required to derive a rigorous bound on the sample cost.

For example, consider the AKLT state on the closed chain with $n=100$ nodes, in which case  $z=2$ and $\gamma\approx 0.350$ \cite{GarcMW13,WeiRA22}. Suppose we want to verify the target AKLT state within infidelity $\epsilon=0.01$ and significance level $\delta=0.01$. We can apply the optimal coloring protocol, and each bond verification protocol can be  constructed from a 4-design [cf. \thref{thm:homoDesignGen}]. 
According to \thref{thm:SpectralGap} with $m=g=2$, $s=1/2$, $\nu_E=2/5$, and $\gamma\approx 0.350$, the spectral gap of the verification operator is bounded from below by 0.0278, and the number of tests required by our protocol is at most $1.66\times 10^4$. 
By contrast, according to \eref{eq:HKSE} with $|E|=n$, the number of  tests required by the  HKSE protocol is about $3.76\times 10^{11}$, which is 22 million times more than our protocol. According to the lower bound in \eref{eq:BHSRElb} with $\kappa=2$, the number of  tests required by the BHSRE protocol is at least $2.32\times 10^9$ (assuming that BHSRE protocol can be generalized to qudit systems; originally,  BHSRE mainly focused on qubit systems), which is 140 thousand times more than our protocol.

Next, consider the AKLT state on the honeycomb lattice with the same number of nodes, in which case  $z=3$ and $\gamma\approx 0.10$ \cite{GarcMW13,WeiRA22}. Now we can apply the optimal coloring protocol as illustrated in \fref{fig:SquareHoney} in the main text, and each bond verification protocol can be  constructed from a 6-design [cf. \thref{thm:homoDesignGen}]. 
According to \thref{thm:SpectralGap} with $m=3$, $g=4$, $s=1/2$, $\nu_E=2/7$, and $\gamma\approx 0.10$,  the spectral gap of the verification operator is bounded from below by $5.8\times 10^{-4}$, and the number of tests required by our protocol is at most $7.9\times 10^5$. By contrast, according to \eref{eq:HKSE} with $|E|=3n/2$, the number of  tests required by the  HKSE protocol
is about $1.6\times 10^{13}$, which is 20 million times more than our protocol.  According to the lower bound in \eref{eq:BHSRElb} with $\kappa=3$, 
the number of  tests required by the BHSRE protocol is at least  $2.9\times 10^{10}$, which is 37 thousand times more than our protocol. 
The advantage of our verification protocol is more dramatic as the system size  and target precision increase.

In addition, our protocol only uses local projective measurements. If the  projectors that compose the Hamiltonian can be measured directly as required in the HKSE protocol, then the bond spectral gap $\nu_E$ in \thref{thm:SpectralGap} can attain the maximum value 1 instead of $2/5$ ($2/7$) for the 1D chain (honeycomb lattice), and the number of  tests required by our protocol can be reduced by a factor of  $2/5$ ($2/7$) for the 1D chain (honeycomb lattice).\\

\subsection{Comparison with \rcite{GluzKEA18}}
In \rcite{GluzKEA18},  Gluza, Kliesch, Eisert, and  Aolita (GKEA) introduced a protocol for verifying fermionic Gaussian states. The focus and scope of applications are very different from the current work. To verify an $L$-mode fermionic Gaussian state within infidelity $\epsilon$ and significance level $\delta$, the sample complexity is about 
\begin{align}
N\leq \biggl\lceil \frac{2L^4 \ln(2/\delta)}{\epsilon^2}\biggr\rceil. 
\end{align}
When the fermionic Gaussian state is the unique ground state of a gapped local Hamiltonian, the sample complexity can be reduced to 
\begin{align}
N\sim \biggl\lceil \frac{L^2(\ln L)^2 \ln(2/\delta)}{2\epsilon^2}\biggr\rceil. 
\end{align}
However, here the constant and scaling behavior with respect to the spectral gap $\gamma$ were not clarified in \rcite{GluzKEA18}. 


\end{document}